\documentclass[11pt]{article}

\RequirePackage{silence}
\usepackage{amsfonts,amsmath,amssymb,amsthm,graphicx}
\usepackage{fullpage}
\usepackage{changepage}
\usepackage{enumerate}
\usepackage{algorithm}
\usepackage{algorithmic}
\usepackage{scalerel}
\usepackage{accents}
\usepackage[margin=1in]{geometry}
\usepackage{hyperref}
\usepackage{mathtools}
\usepackage{bbm}
\usepackage{natbib}
\usepackage{csquotes}
\usepackage{comment}
\usepackage{thmtools}
\usepackage{thm-restate}
\usepackage{caption}
\usepackage{capt-of}
\usepackage[capitalize]{cleveref}

\usepackage{tikz}
\usetikzlibrary{arrows.meta}
\usetikzlibrary{matrix}
\usepackage{pgfplots}
\pgfplotsset{compat=1.16}

\usetikzlibrary{arrows}

\usepackage[caption=false]{subfig}

\theoremstyle{plain}
\newtheorem{theorem}{Theorem}[section]
\newtheorem{lemma}[theorem]{Lemma}

\newtheorem{proposition}[theorem]{Proposition}
\newtheorem{corollary}[theorem]{Corollary}% reset theorem numbering for each chapter

\theoremstyle{plain}
\newtheorem{definition}{Definition}[section] % definition numbers are dependent on theorem numbers
\newtheorem{example}[definition]{Example}
\newtheorem{assumption}[definition]{Assumption}

\allowdisplaybreaks

\let\oldparagraph\paragraph
\renewcommand{\paragraph}[1]{\oldparagraph{#1.}}

\newcommand{\xhdr}[1]{\vspace{2mm} \noindent{\bf #1}}

\newcommand{\cdw}{{\rm{CDW}}}

\newcommand{\bspa}{{\rm BSPA}}
\newcommand{\brev}{{\rm BRev}}
\newcommand{\srev}{{\rm SRev}}
\newcommand{\vcg}{{\rm VCG}}
\newcommand{\val}{{v}}
\newcommand{\rev}{{\rm OPT}}
\newcommand{\OPT}{{\rm OPT}}
\newcommand{\wel}{{\rm WEL}}
\newcommand{\core}{{\rm CORE}}
\newcommand{\auc}{{\rm AUC}}

\newcommand{\itemsubset}{S}
\newcommand{\mech}{\mathcal{M}}
\newcommand{\const}{\mathcal{C}}
\newcommand{\constna}{\const_{n,\alpha}}

\newcommand{\quant}{q}

\newcommand{\primed}{^\dagger}
\newcommand{\doubleprimed}{^\ddagger}

\newcommand{\event}{\mathcal{E}}

\newcommand{\QUANT}{{\rm QuantCDW}}
\newcommand{\RANDQ}{{\rm RandQCDW}}
\newcommand{\RANDQPlus}{{\rm RandQCDW+}}

\newcommand{\naturals}{\mathbb{N}}
\newcommand{\reals}{\mathbb{R}}

\newcommand{\dist}{F}
\newcommand{\cdf}{F}
\newcommand{\pdf}{f}
\newcommand{\virtual}{\varphi}
\newcommand{\hazard}{h}

\newcommand{\quantMatrix}{Q}
\newcommand{\contribution}{\xi}

\newcommand{\Payoff}[2][]{\text{\bf Payoff}\ifthenelse{\not\equal{}{#1}}{_{#1}}{}\!\left[{\def\givenn{\middle|}#2}\right]}

\newcommand{\revcurve}{R}

\DeclareMathOperator{\argmax}{argmax}

\newcommand{\given}{\,\mid\,}

\newcommand{\prob}[2][]{\text{\bf Pr}\ifthenelse{\not\equal{}{#1}}{_{#1}}{}\!\left[{\def\givenn{\middle|}#2}\right]}
\newcommand{\expect}[2][]{\text{\bf E}\ifthenelse{\not\equal{}{#1}}{_{#1}}{}\!\left[{\def\givenn{\middle|}#2}\right]}

\newcommand{\rbr}[1]{\left(\,#1\,\right)}

\newcommand{\cbr}[1]{\left\{\,#1\,\right\}}

\newcommand{\tparen}{\big}
\newcommand{\tprob}[2][]{\text{\bf Pr}\ifthenelse{\not\equal{}{#1}}{_{#1}}{}\tparen[{\def\given{\tparen|}#2}\tparen]}
\newcommand{\texpect}[2][]{\text{\bf E}\ifthenelse{\not\equal{}{#1}}{_{#1}}{}\tparen[{\def\given{\tparen|}#2}\tparen]}

\newcommand{\sprob}[2][]{\text{\bf Pr}\ifthenelse{\not\equal{}{#1}}{_{#1}}{}[#2]}
\newcommand{\sexpect}[2][]{\text{\bf E}\ifthenelse{\not\equal{}{#1}}{_{#1}}{}[#2]}

\newcommand{\dd}{{\,\mathrm d}}

\newcommand{\abs}[1]{\left| #1 \right |}

\newcommand{\indicator}[2][]{\mathbbm{1}\ifthenelse{\not\equal{}{#1}}{_{#1}}{}\!\left[{\def\givenn{\middle|}#2}\right]}

\newcommand{\plus}[1]{\left({\def\givenn{\middle|}#1}\right)^+}

\begin{document}

\title{Competition Complexity in Multi-Item Auctions:\\
Beyond VCG and Regularity}

\author{Hedyeh Beyhaghi\thanks{University of Massachusetts Amherst. \url{hbeyhaghi@umass.edu}}
\and Linda Cai\thanks{University of California, Berkeley. Email: \url{tcai@berkeley.edu}}
\and Yiding Feng\thanks{Hong Kong University of Science and Technology. Email: \url{ydfeng@ust.hk}}
\and Yingkai Li\thanks{National University of Singapore. Email: \url{yk.li@nus.edu.sg}}
\and S. Matthew Weinberg\thanks{Princeton University. Email: \url{smweinberg@princeton.edu}}
}
\date{}

\maketitle

\begin{abstract}
We quantify the value of the monopoly's bargaining power in terms of competition complexity---that is, the number of additional bidders the monopoly must attract in simple auctions to match the expected revenue of the optimal mechanisms \citep[c.f.,][]{BK-96,EFFTW-17}---within the setting of multi-item auctions. 
We show that for simple auctions that sell items separately, the competition complexity is $\Theta(\frac{n}{\alpha})$ in an environment with $n$ original bidders under the slightly stronger assumption of $\alpha$-strong regularity, in contrast to the standard regularity assumption in the literature, which requires $\Omega(n \cdot \ln \frac{m}{n})$ additional bidders \citep{FFR-18}. This significantly reduces the value of learning the distribution to design the optimal mechanisms, especially in large markets with many items for sale.
For simple auctions that sell items as a grand bundle, we establish a constant competition complexity bound in a single-bidder environment when the number of items is small or when the value distribution has a monotone hazard rate.
Some of our competition complexity results also hold when we compete against the first best benchmark (i.e., optimal social welfare).
\end{abstract}
\newpage

\section{Introduction}
\label{sec:intro}

To maximize profit, a monopoly seller can either run a simple auction among potential buyers or restrict buyer access to gain more control over bargaining and negotiation. In their seminal work, \citet{BK-96} show that for single-item auctions under mild regularity assumptions, the latter approach offers no advantage: the seller can achieve at least the same revenue as the optimal mechanism by simply attracting one additional buyer and running a second-price auction.\footnote{The second-price auction awards the item to the highest bidder and charges them a price equal to the second-highest bid.} This result also suggests that, given limited resources, a firm may benefit more from expanding its market while using simple auctions rather than fine-tuning the selling mechanism based on detailed factors of the environment.

In many practical applications, the seller typically has multiple items to sell, and buyers often have combinatorial values for these items. In such settings, the Bayesian optimal mechanism becomes significantly more complex, even when selling just two items to a single bidder. Specifically, there are known examples where the optimal mechanism may sell each item separately, sell both items as a grand or nested bundle, offer a ``discount price'' for the bundle, sell lotteries, and other possibly counterintuitive structure \citep[see, e.g.,][]{MM-88,MM-92,HN-12,DDT-17,HR-15,HH-21,yan-22}. 
Given the complexity of designing the optimal mechanism, generalizing the result of \citet{BK-96} to combinatorial auctions becomes even more appealing. This is because the cost of attracting additional buyers is likely to be similar in both settings, while optimizing the mechanisms in combinatorial auctions is considerably more challenging.

Building on this motivation, \citet{EFFTW-17} introduce and initiate the study of \emph{competition complexity} in multi-item settings. Specifically, the competition complexity of a simple auction is defined as the number of additional bidders required to ensure that its expected revenue exceeds the Bayesian optimal revenue from the original number of bidders. Under this definition, in the single-item setting, the competition complexity of the second-price auction is 1 for regular bidders \citep{BK-96}. 
In multi-item settings, one classic simple mechanism that does not require detailed knowledge of the environment is the Vickrey-Clarke-Groves (VCG) auction.\footnote{VCG auction is the mechanism that implements the efficient allocation. In the special case with additive values, VCG auction sells the items separately using second price auctions.} 
For regular bidders,\footnote{In multi-item settings, we say a bidder is regular (resp.\ MHR) if her valuation for each item is realized from a regular (resp.\ MHR) valuation distribution. Similarly, we say a item is regular (resp.\ MHR) if its corresponding distribution is regular (resp.\ MHR). See \cref{sub:dist_assumption} for formal definitions.} when the buyers have additive valuations, the tight bound on the competition complexity for the VCG auction has been identified in \citet{BW-19} and \citet{DRWX-24}. Specifically, for regular bidders, the VCG auction has competition complexity $\Theta(n\log({m}/{n}))$ when the number of bidders $n$ is smaller than the number of items $m$.

Compared to the competition complexity in the single-item setting, the competition complexity of the VCG auction in multi-item settings depends on both the number of bidders, $n$, and the number of items, $m$. While these dependencies are theoretically necessary, it is unclear whether they are required in more realistic scenarios (e.g., natural distribution subclasses) or can be avoided by considering other mechanisms. To shed light on this, let us consider the following two examples.

\begin{example}[Competition complexity for exponential distribution]
\label{example:intro:exp}
    Suppose there is a single bidder and $m$ items. The value for each item is realized i.i.d.\ from the exponential distribution~$\dist$ with CDF $\dist(\val) = 1 - e^{-\val}$ for every $\val\in[0, \infty)$. 
    For this instance, the competition complexity of the VCG auction is 3.\footnote{We remark that analytically characterizing the revenue of the Bayesian optimal mechanism or other commonly used mechanisms in the multi-item setting for a specific distribution is known to be challenging and often requires technically non-trivial analysis \citep[e.g.,][]{DDT-17,DRWX-24}. In both \Cref{example:intro:exp,example:intro:eqr}, we numerically evaluate the performance of the VCG auction and the bundle-based second-price auction, which are plotted in \Cref{fig:intro:example}. 
    Except for $\Omega(\log(m))$ in \Cref{example:intro:eqr}, all other competition complexities reported in these two examples are based on this numerical evaluation.}
    % See Appendix XYZ for the detailed discussion of our numerical evaluation.
    In other words, the expected revenue of the VCG auction with 4 bidders, each having values drawn from the exponential distribution, exceeds the expected revenue of the Bayesian optimal mechanism when there is only a single bidder. 
    Moreover, there exists another simple mechanism---selling all items as a grand bundle using the second-price auction---that has the competition complexity of 2.
    Also see \Cref{fig:example-exp} for an illustration.
\end{example}

\begin{example}[Competition complexity for equal-revenue distribution]
\label{example:intro:eqr}
    Suppose there is a single bidder and $m$ items. The value for each item is realized i.i.d.\ from the equal-revenue distribution~$\dist$ with CDF $\dist(\val) = 1 - \frac{1}{\val}$ for every $\val\in[1, \infty)$. 
    For this instance, the competition complexity of the VCG auction is $\Theta(\log(m))$. In contrast, the competition complexity of the second-price auction for the grand bundle is 1.
    In other words, the expected revenue from selling all items as a grand bundle using the second-price auction with 2 bidders, each having values drawn from the equal-revenue distribution, exceeds the expected revenue of the Bayesian optimal mechanism when there is only a single bidder. Also see \Cref{fig:example-eqr} for an illustration.
\end{example}

\begin{figure}
    \centering
    \subfloat[Exponential distribution.]{
\begin{tikzpicture}

\begin{axis}[
    axis lines = middle,  % Draw the axis lines through the middle
    xlabel = {$m$},  % Label for the x-axis
    ylabel = {Rev},  % Label for the y-axis
    % legend style = {at={(0.5,-0.1)}, anchor=north, legend columns=1},  % Legend below the plot
    legend style = {at={(0.15,1)}, anchor=north west, draw=none, fill=none, font=\small},  % Adjusted legend position
    % legend style = {at={(0,1)}, anchor=north west, draw=none, fill=none, font=\small},  % Position and style of legend
    xtick = {0,1,40},  % Set the x-axis ticks manually
    xmin = 0, xmax = 40,  % Set the x-axis range to ensure ticks display correctly
    ytick = {0},
    yticklabels={0},
    width=8cm,  % Set the width of the plot
    height=6cm   % Set the height of the plot
]

% Draw the smooth curve and label it in the legend
\addplot[thick, smooth, black, line width=0.5mm] coordinates {
(1, 1.3682852097302007)
(2, 2.172469460764839)
(3, 3.0957404465251916)
(4, 4.018458798617175)
(5, 5.059263856399517)
(6, 5.997296618166068)
(7, 6.960127986147769)
(8, 8.028511635386536)
(9, 9.024327197268455)
(10, 10.0212144167913)
(11, 11.007457768128951)
(12, 12.024478328847515)
(13, 13.04964453030675)
(14, 14.03558156560873)
(15, 14.961461249550673)
(16, 15.975249541616588)
(17, 16.982974033976987)
(18, 18.053419340908)
(19, 19.076359829762463)
(20, 19.943719672632703)
(21, 20.856930084522002)
(22, 21.959405368124628)
(23, 22.959105508818308)
(24, 24.04648544048241)
(25, 24.94999767696281)
(26, 26.10489987009769)
(27, 26.974810113685653)
(28, 28.012447567467213)
(29, 28.959570281324314)
(30, 30.024468316363503)
(31, 31.01055906473247)
(32, 31.906424279737426)
(33, 32.98145651767206)
(34, 34.05655488699192)
(35, 34.99878747200814)
(36, 35.941567422954954)
(37, 36.961120421129564)
(38, 38.00231487359635)
(39, 39.039816577934644)
(40, 40.07786365643654)};
\addlegendentry{$\cdw_1$}  % Add a legend entry for the curve

% Draw the smooth curve and label it in the legend
\addplot[thick, smooth, blue, line width=0.5mm] coordinates {(1, 1.8539218460276783)
(2, 2.83962172951938)
(3, 3.822591812116968)
(4, 4.820208151274721)
(5, 5.8150972704153565)
(6, 6.843268893912943)
(7, 7.8066028468573645)
(8, 8.82354531845564)
(9, 9.814087278682214)
(10, 10.838193054995106)
(11, 11.801113731433198)
(12, 12.830983566884248)
(13, 13.826610557681349)
(14, 14.839531101902182)
(15, 15.791875405111746)
(16, 16.83381372419662)
(17, 17.80521309769748)
(18, 18.791887679515952)
(19, 19.84708854416819)
(20, 20.799507283556025)
(21, 21.813305554035736)
(22, 22.80312907803259)
(23, 23.780752855777212)
(24, 24.81187754205159)
(25, 25.81195121849545)
(26, 26.84692340877675)
(27, 27.82638197699361)
(28, 28.83234160759875)
(29, 29.748095762254664)
(30, 30.807969264318817)
(31, 31.81544301708546)
(32, 32.77452753414314)
(33, 33.891553463971256)
(34, 34.84092637391482)
(35, 35.84486727456337)
(36, 36.79513541595873)
(37, 37.75993776822781)
(38, 38.8458466527415)
(39, 39.783434288511955)
(40, 40.844127119606306)};
\addlegendentry{$\bspa_3$}  % Add a legend entry for the curve

% Draw the smooth curve and label it in the legend
\addplot[thick, smooth, blue, dashed, line width=0.5mm] coordinates {(1, 1.5108913819380843)
(2, 2.2631999564831737)
(3, 3.066834883098664)
(4, 3.8979441625581397)
(5, 4.769429491017126)
(6, 5.665543198041163)
(7, 6.517349768056321)
(8, 7.4327658335635105)
(9, 8.33947071435049)
(10, 9.261972729989822)
(11, 10.170044273458641)
(12, 11.055984069256311)
(13, 12.017066610817341)
(14, 12.944718759119507)
(15, 13.83041047770058)
(16, 14.763017406211613)
(17, 15.685535824748277)
(18, 16.61193717721951)
(19, 17.60163564681068)
(20, 18.445543785256753)
(21, 19.432914438969277)
(22, 20.33495781455149)
(23, 21.31236155667737)
(24, 22.249512748271428)
(25, 23.194814309643185)
(26, 24.20124444994268)
(27, 25.11206123304587)
(28, 26.079436698454312)
(29, 26.938461138949062)
(30, 27.91197211011441)
(31, 28.89466864807837)
(32, 29.74147029285811)
(33, 30.777366509748042)
(34, 31.74440993659137)
(35, 32.66950161367302)
(36, 33.634640471015516)
(37, 34.558326091529146)
(38, 35.53190950726476)
(39, 36.493505085837505)
(40, 37.48837354497289)};
\addlegendentry{$\bspa_2$}  % Add a legend entry for the curve

% Draw the smooth curve and label it in the legend
\addplot[thick, smooth, red, line width=0.5mm] coordinates {(1, 2.0956587414894927)
(2, 3.1727196838935123)
(3, 4.254231031337463)
(4, 5.332414836279613)
(5, 6.433602663399976)
(6, 7.529883811015701)
(7, 8.58419124445429)
(8, 9.658456171693329)
(9, 10.751641159730502)
(10, 11.84036760876)
(11, 12.886542607963072)
(12, 14.006424086629202)
(13, 15.116509885085401)
(14, 16.196480549097885)
(15, 17.234985815344768)
(16, 18.284963274928245)
(17, 19.439392708398934)
(18, 20.49862518790427)
(19, 21.59700116423063)
(20, 22.653979045202014)
(21, 23.75969969993831)
(22, 24.82372477928131)
(23, 25.85032469450051)
(24, 27.0262648035682)
(25, 28.12704112276007)
(26, 29.177698035432996)
(27, 30.231996519844888)
(28, 31.380917132062784)
(29, 32.383230216374045)
(30, 33.51430316228393)
(31, 34.56263228206498)
(32, 35.649533239661)
(33, 36.788922294595245)
(34, 37.911759937795836)
(35, 38.945113418829486)
(36, 39.95965008207006)
(37, 41.021068020533455)
(38, 42.17255273925109)
(39, 43.231338097563714)
(40, 44.3935623774543)};
\addlegendentry{$\vcg_4$}  % Add a legend entry for the curve

% Draw the smooth curve and label it in the legend
\addplot[thick, smooth, red, dashed, line width=0.5mm] coordinates {(1, 1.8539218460276783)
(2, 2.675215506081067)
(3, 3.500354203945881)
(4, 4.330796765550376)
(5, 5.1635957200906875)
(6, 6.017872965369171)
(7, 6.817183458393404)
(8, 7.669109193811497)
(9, 8.495468508017701)
(10, 9.332782469593745)
(11, 10.15352606311076)
(12, 11.019070185047877)
(13, 11.844490195234796)
(14, 12.69755954889937)
(15, 13.473564522388742)
(16, 14.320089176533976)
(17, 15.14337058667299)
(18, 15.999068557025486)
(19, 16.847416595549777)
(20, 17.668815092770014)
(21, 18.480075183355435)
(22, 19.320951196898378)
(23, 20.127089884796685)
(24, 21.00309870446764)
(25, 21.859214942152438)
(26, 22.695879734855307)
(27, 23.495282749342213)
(28, 24.38804899489288)
(29, 25.09488163240524)
(30, 26.007808841380452)
(31, 26.787441823873934)
(32, 27.641026598319876)
(33, 28.5218457380169)
(34, 29.418831977305086)
(35, 30.189601107346213)
(36, 30.98446679512331)
(37, 31.798684182378604)
(38, 32.68439233392313)
(39, 33.48931541868708)
(40, 34.377977378711556)};
\addlegendentry{$\vcg_3$}  % Add a legend entry for the curve

\end{axis}

\end{tikzpicture}
\label{fig:example-exp}
}~~~~
    \subfloat[Equal-revenue distribution.]{
\begin{tikzpicture}

\begin{axis}[
    axis lines = middle,  % Draw the axis lines through the middle
    xlabel = {$m$},  % Label for the x-axis
    ylabel = {Rev},  % Label for the y-axis
    % legend style = {at={(0.5,-0.1)}, anchor=north, legend columns=1},  % Legend below the plot
    legend style = {at={(0.15,1)}, anchor=north west, draw=none, fill=none, font=\small},  % Adjusted legend position
    % legend style = {at={(0,1)}, anchor=north west, draw=none, fill=none, font=\small},  % Position and style of legend
    xtick = {0,1,40},  % Set the x-axis ticks manually
    xmin = 0, xmax = 40,  % Set the x-axis range to ensure ticks display correctly
    ytick = {0},
    yticklabels={0},
    width=8cm,  % Set the width of the plot
    height=6cm   % Set the height of the plot
]

% Draw the smooth curve and label it in the legend
\addplot[thick, smooth, black, line width=0.5mm] coordinates {(1, 1.0)
(2, 3.012947383862324)
(3, 5.483601519090033)
(4, 8.317515418506876)
(5, 11.420878407350248)
(6, 14.758202866711029)
(7, 18.151767411958993)
(8, 21.813510361237416)
(9, 25.492280624779234)
(10, 29.19104440326029)
(11, 33.43906639818819)
(12, 37.29547812629904)
(13, 41.02385281190555)
(14, 45.47856788630751)
(15, 49.704338401395)
(16, 53.85566247993715)
(17, 58.32759388004927)
(18, 62.87953555029362)
(19, 67.3116806491388)
(20, 72.23996611794003)
(21, 76.43897167021542)
(22, 81.65873569716729)
(23, 85.81135045224627)
(24, 90.68077063403656)
(25, 95.42982224021891)
(26, 99.85016136229808)
(27, 104.99598030720063)
(28, 109.8075086856517)
(29, 115.06339039214416)
(30, 119.69528539880856)
(31, 124.84953307685193)
(32, 129.79473481519108)
(33, 134.89788075604483)
(34, 140.37071487755782)
(35, 145.2273885806756)
(36, 150.3123656861908)
(37, 155.00871588227452)
(38, 160.3426105345174)
(39, 165.5524318370346)
(40, 171.22839482852123)};
\addlegendentry{$\cdw_1$}  % Add a legend entry for the curve

% Draw the smooth curve and label it in the legend
\addplot[thick, smooth, blue, line width=0.5mm] coordinates {(1, 3.012466614637677)
(2, 5.843250929723265)
(3, 9.131933310222879)
(4, 12.701305197847732)
(5, 16.566017539705925)
(6, 20.609883891991426)
(7, 24.680907463311307)
(8, 28.856215035723903)
(9, 33.48410101734121)
(10, 37.68284101396826)
(11, 42.18913878842266)
(12, 46.803752073834985)
(13, 51.549045897766554)
(14, 56.271714173472375)
(15, 61.37575927673981)
(16, 66.36978000586049)
(17, 70.94053302299204)
(18, 76.23634721816038)
(19, 81.27955760932825)
(20, 86.69185602759669)
(21, 91.60328276461311)
(22, 97.32186727034758)
(23, 102.52087635320716)
(24, 107.40759279305732)
(25, 113.05437751085661)
(26, 119.10540892798478)
(27, 123.7965871413247)
(28, 129.91290397503204)
(29, 134.95234058257967)
(30, 140.91589161277585)
(31, 146.68155062133476)
(32, 151.34768318351496)
(33, 157.4636133573102)
(34, 163.5976455401868)
(35, 168.88341808293123)
(36, 174.32749050337003)
(37, 179.9733432221052)
(38, 186.10297475991442)
(39, 191.70832561074934)
(40, 197.76563853872)};
\addlegendentry{$\bspa_2$}  % Add a legend entry for the curve

% Draw the smooth curve and label it in the legend
\addplot[thick, smooth, red, line width=0.5mm] coordinates {(1, 4.993432874815371)
(2, 8.966709399069515)
(3, 12.915022233344695)
(4, 16.96786453850246)
(5, 20.900881604557025)
(6, 25.053129609669373)
(7, 29.00627303784104)
(8, 33.095651495151046)
(9, 37.00881327874097)
(10, 40.95717600393495)
(11, 45.0314294362327)
(12, 48.99412175090896)
(13, 52.88953454936329)
(14, 56.86843721942517)
(15, 61.1934857847612)
(16, 65.02781291185316)
(17, 69.18680170342093)
(18, 72.84615229007083)
(19, 76.93321449959087)
(20, 81.36180678005314)
(21, 84.82707045324031)
(22, 89.11269628632257)
(23, 93.12055171470773)
(24, 97.01412973278367)
(25, 100.90978081124598)
(26, 104.85701299440427)
(27, 108.66559226053971)
(28, 113.05209442761372)
(29, 117.12073477730019)
(30, 120.98618077011292)
(31, 125.06127057045263)
(32, 128.74389849839983)
(33, 132.93628686378995)
(34, 136.99933059644474)
(35, 141.16544180127795)
(36, 145.47213956460354)
(37, 148.87134303118856)
(38, 153.41191029611235)
(39, 156.8051002488904)
(40, 161.08947439018775)};
\addlegendentry{$\vcg_4$}  % Add a legend entry for the curve

\end{axis}

\end{tikzpicture}
\label{fig:example-eqr}
}
\caption{The revenue comparison between different mechanisms and benchmark for selling $m\in[40]$ i.i.d.\ items to a single bidder. $\cdw_n$ is a duality-based upper bound of the Bayesian optimal revenue for $n$ bidders, introduced in \citet{CDW-16}. $\vcg_n$ and $\bspa_n$ are the expected revenue from the VCG auction and the second-price auction for the grand bundle for $n$ bidders, respectively. All curves are evaluated numerically.}
    \label{fig:intro:example}
\end{figure}
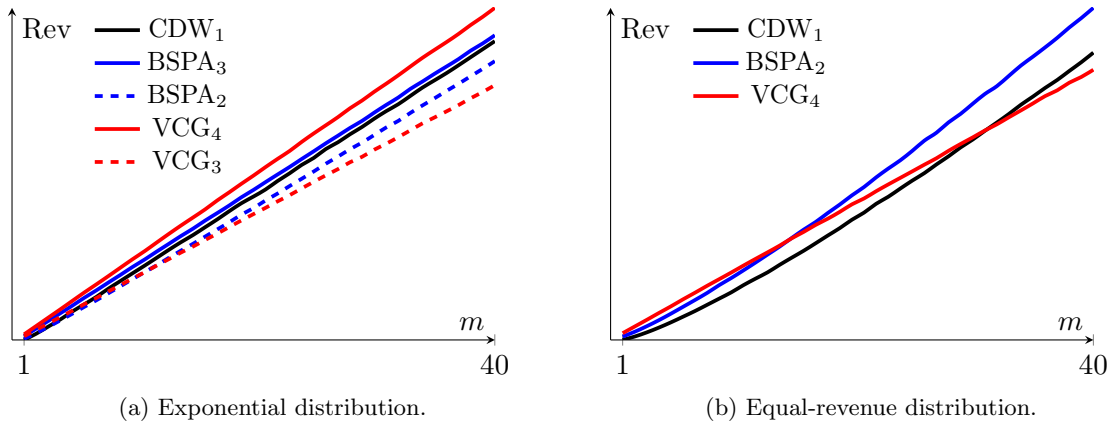

In both \Cref{example:intro:exp,example:intro:eqr}, we focus on the single-bidder instances, whose Bayesian optimal mechanisms are already complicated. Both the exponential distribution and the equal-revenue distribution are paradigmatic in the mechanism design literature. In \Cref{example:intro:exp}, the value distribution is exponential, which is not only regular but also satisfies a stronger property known as the monotone hazard rate (MHR) condition. Since the theoretically tight bound for the competition complexity of the VCG auction is derived for the family of regular distributions, it is natural to ask the following question:
\begin{displayquote}
   (Question 1) \emph{Can we obtain improved competition complexity bounds of the VCG auction under stronger distributional assumptions than regularity?}
\end{displayquote}
\Cref{example:intro:eqr} considers the equal-revenue distribution, which is exactly the distribution used in \citet{BW-19} to prove the competition complexity lower bound for the VCG auction. For the VCG auction, $\Theta(\log(m))$ additional bidders is necessary to outperform the Bayesian optimal revenue for a single bidder. However, as we illustrate in this example, there exists another simple mechanism---selling all items as a grand bundle using the second-price auction---which has the competition complexity of 1 for this instance. Motivated by this, it is natural to ask the following question:
\begin{displayquote}
   (Question 2) \emph{Can we obtain improved competition complexity bounds by considering other simple mechanisms beyond the VCG auction?}
\end{displayquote}
In this paper, we answer Question 1 in the affirmative. Additionally, we make conceptually novel and technically non-trivial progress in partially answering Question 2. In the following subsection, we provide a summary of our results and proof techniques.

\subsection{Our Contributions and Techniques} \label{sec:intro:contribution}

In this subsection we explain the contributions and techniques of this paper.

Our paper focuses on revenue maximization in a classic multi-item setting where bidders are ex ante symmetric and have additive valuation functions. Specifically, each item $j$ is associated with a valuation distribution $\dist_j$. Each bidder $i$ has a private value $\val_{ij}$ realized from distribution $\dist_j$ for item $j$, and his valuation (i.e., willingness-to-pay) for item subset $\itemsubset$ is $\sum_{j\in\itemsubset}\val_{ij}$. The goal of the seller is to design truthful mechanisms to maximize the expected revenue (i.e., total expected payment from all bidders).\footnote{By the revelation principle \citep{mye-81}, restricting to truthful (or more precisely, Bayesian incentive compatible) mechanisms is without loss of generality.}

First, we introduce the mechanisms and benchmarks studied in the this paper. (See \Cref{sec:prelim} for the formal definitions.)
\begin{itemize}
    \item {\sf Bayesian optimal mechanism} {(\rev)}: the revenue-optimal mechanism among all mechanisms. It could be complicated even for selling two items to a single bidder. We refer to its expected revenue as the Bayesian optimal revenue.
    \item {\sf Duality benchmark} {(\cdw)}: viewing the Bayesian optimal mechanism as a (primal) linear program, this benchmark is a constructed feasible dual solution of the dual program developed in \citet{CDW-16}. Due to weak duality of the linear program, it is an upper bound of the Bayesian optimal revenue. It is one of the most commonly used benchmarks in the multi-item mechanism design (e.g. \citep{CZ-17,EFFTW-17,BCWZ-17,DW-17,CS-21}).
    \item {\sf Welfare benchmark} {(\wel)}: the expected optimal welfare for the instance. It is an upper bound of the Bayesian optimal revenue, and also known as the first-best benchmark.
    \item {\sf Vickery-Clarke-Groves (VCG) auction} {(\vcg)}: this is a prior-independent mechanism. Since bidders have additive valuation functions, it sells each item separately using the second-price auction.
    \item {\sf Bundle-based second-price auction} {(\bspa)}: this is a prior-independent mechanism. It sells all items as a grand bundle using the second-price auction.
\end{itemize}
All three mechanisms and two benchmarks form the hierarchy $\vcg,\bspa\leq \rev\leq \cdw \leq \wel$, with only one incomparable pair ``{\vcg} vs.\ {\bspa}''. See \Cref{fig:mechanism diagram} for a graphical illustration.
To emphasize the dependence on the number of bidders and shed lights on the competition complexity, we denote the Bayesian optimal revenue, the expected value of duality benchmark and welfare benchmark, and the expected revenue of the VCG auction and bundle-based second-price auction for $n$ bidders as $\rev_n$, $\cdw_n$, $\wel_n$, $\vcg_n$, and $\bspa_n$, respectively.

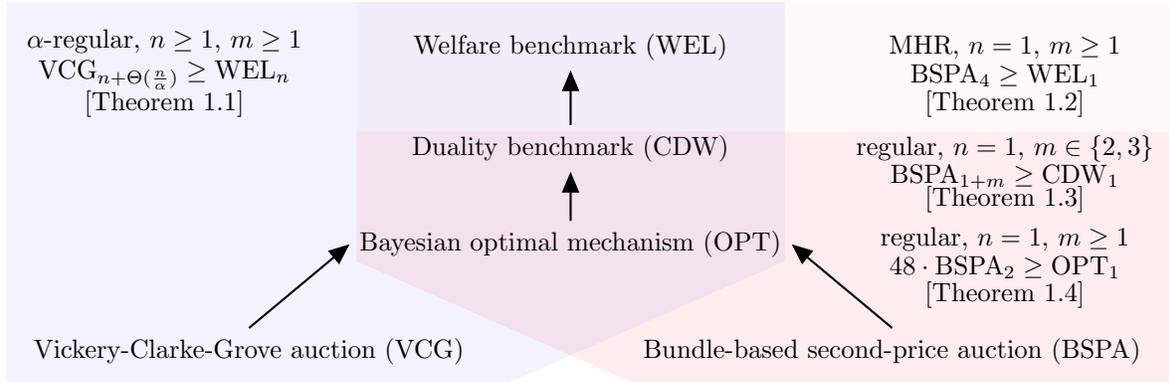
\begin{figure}
    \centering
    \begin{tikzpicture}[scale = 1.5,every node/.style={font=\small}]

    \fill[red, fill opacity = 0.05] (5.4,-0.5) -- (0.6, -0.5) -- (-1.9, 0.6) -- (-1.9, 1.75) -- (5.4, 1.75) -- cycle;

    \fill[red, fill opacity = 0.025] (5.4,-0.5) -- (0.6, -0.5) -- (-1.9, 0.6) -- (-1.9, 2.9) -- (5.4, 2.9) -- cycle;
    
    \fill[blue, fill opacity = 0.05] (-5.0,-0.5) -- (-0.6, -0.5) -- (1.9, 0.6) -- (1.9, 2.9) -- (-5.0, 2.9) -- cycle;

    % \node(thm1a) at (-3., 2.5) {adding $\Theta(\frac{n}{\alpha})$ bidders};
    % \node(thm1b) at (-3., 2.25) {$\alpha$-regular, $n\geq 1$, $m\geq 1$};
    % \node(thm1a) at (-3., 2.25) {$\Theta(\frac{n}{\alpha})$ additional bidders};
    \node(thm1a) at (-3.6, 2.25) {$\vcg_{n + \Theta(\frac{n}{\alpha})} \geq \wel_{n}$};
    \node(thm1b) at (-3.6, 2.55) {$\alpha$-regular, $n\geq 1$, $m\geq 1$};
    \node(thm1c) at (-3.6, 2) {[\Cref{thm:intro:vcg:alpha-regular}]};

    % \node(thm2a) at (3.25, 2.25) {adding $3$ bidders};
    \node(thm2a) at (3.85, 2.25) {$\bspa_{4} \geq \wel_1$};
    \node(thm2b) at (3.85, 2.5) {MHR, $n = 1$, $m\geq 1$};
    \node(thm2c) at (3.85, 2) {[\Cref{thm:intro:bspa:mhr}]};

    % \node(thm3a) at (3.25, 1.6) {adding $3$ bidders};
    % \node(thm3b) at (3.25, 1.35) {regular, $n = 1$, $m = 3$};
    \node(thm3a) at (3.85, 1.35) {$\bspa_{1 + m} \geq \cdw_1$};
    \node(thm3b) at (3.85, 1.6) {regular, $n = 1$, $m \in \{2, 3\}$};
    \node(thm3c) at (3.85, 1.15) {[\Cref{thm:intro:bspa:regular:3 items}]};

    % \node(thm4a) at (3.25, 0.8) {adding $1$ bidders, $\frac{1}{48}$-approx.};
    % \node(thm4b) at (3.25, 0.55) {regular, $n = 1$, $m \geq 1$};
    \node(thm4a) at (3.85, 0.55) {$48\cdot \bspa_{2} \geq \rev_{1}$};
    \node(thm4b) at (3.85, 0.8) {regular, $n = 1$, $m \geq 1$};
    \node(thm4c) at (3.85, 0.3) {[\Cref{thm:intro:bspa:regular:adding one bidder}]};
    
    % \node(thm3) at (3.45, 1.5) {\Cref{prop:intro:bspa}};

    % \node(thm3) at (3.45, 1.2) {\Cref{thm:comp_complexity_3_item}};

    \node(wel) at (0, 2.5) {{Welfare benchmark} ({\wel})};

    \node(cdw) at (0, 1.6) {{Duality benchmark} ({\cdw})};
    
    \node(rev) at (0, 0.75) {{Bayesian optimal mechanism} ({\rev})};

    \node(vcg) at (-2.85, -0.2) {{Vickery-Clarke-Grove auction} ({\vcg})};
    
    \node(bspa) at (2.85, -0.2) {{Bundle-based second-price auction} ({\bspa})};
    
    \draw[thick, <-, >=triangle 45] (wel.south) to node[anchor = -30] {} (cdw.north);
    
    \draw[thick, <-, >=triangle 45] (cdw.south) to node[anchor = -30] {} (rev.north);
    
    \draw[thick, <-, >=triangle 45] (rev.west) to node[anchor = -30] {} (vcg.north);
    
    \draw[thick, <-, >=triangle 45] (rev.east) to node[anchor = -30] {} (bspa.north);

    % \node(n1) at (2.25, 0.7) {{SS} ({BOM})};
    % \node(n2) at (0, 0) {{BayesianOptimalUniformReserve} ($BOUR$)};
    % \node(n3) at (4.5, 0) {{BayesianOptimalSequentialPricing} ({BOSP})};
    % \node(n4) at (2.25, -0.7) {{BayesianOptimalUniformPricing} ({BOUP})};
    % \draw[thick, <-, >=triangle 45] (n1.west) to node[anchor = -30] {$\star$} (n2.north);
    % \draw[thick, <-, >=triangle 45] (n1.east) to (n3.north);
    % \draw[thick, <-, >=triangle 45] (n1.south) to node[left] {$\star$} (n4.north);
    % \draw[thick, <-, >=triangle 45] (n2.south) to (n4.west);
    % \draw[thick, <-, >=triangle 45] (n3.south) to node[anchor = 150] {$\star$} (n4.east);
\end{tikzpicture}
    \caption{A Hasse diagram of the multi-item mechanisms and benchmarks (i.e., an arrow ``$\mech_{1} \to \mech_{2}$'' means the latter $\mech_{2}$ has a weakly higher payoff than the former $\mech_{1}$). The approximation ratio and competition complexity results of this paper 
    (\Cref{thm:intro:vcg:alpha-regular,thm:intro:bspa:mhr,thm:intro:bspa:regular:3 items,thm:intro:bspa:regular:adding one bidder})
    are illustrated as the shaded areas, which specify the studied mechanism and its compared benchmarks. The subscript index denotes the number of bidders in the mechanism.}
    \label{fig:mechanism diagram}
\end{figure}

\xhdr{Improved competition complexity of the VCG auction.} In the first part of this paper, we focus on the VCG auction and study its competition complexity within the class of $\alpha$-strongly regular distributions. The concept of strong regularity was first introduced by \citet{CR-14} and has since been explored in various contexts within the mechanism design literature \citep[e.g.,][]{CR-17,SS-19,GPZ-21,ABB-22,ABB-23,GPT-23}. Specifically, $\alpha$-strongly regular distributions serve as an interpolation between MHR distributions and regular distributions, with $\alpha = 1$ and $\alpha = 0$ representing the two extreme cases, respectively. 

Recall that both the exponential distribution and the equal-revenue distribution, as discussed in \Cref{example:intro:exp,example:intro:eqr}, are paradigmatic (and often provable ``worst-case" instances) within the classes of MHR and regular distributions, respectively. Inspired by the performance of the VCG auction in these two examples, along with the established tight bound of competition complexity, $\Theta(n\log(\frac{m}{n}))$, for regular distributions \citep{BW-19}, it becomes natural and intriguing to investigate the competition complexity of the VCG auction for MHR distributions and, more broadly, for $\alpha$-strongly regular distributions for every $\alpha \in (0, 1]$. 

As the first main result of this paper, \Cref{thm:intro:vcg:alpha-regular} characterizes the \emph{tight} competition complexity of the VCG auction for $\alpha$-strongly regular distributions given every $\alpha\in(0, 1]$. This helps us answer the Question 1 imposed earlier in our introduction. (The formal theorem statements can be found in \Cref{thm:welfare,thm:tightness_against_revenue}.)

\begin{theorem}[Competition complexity of {\vcg} for $\alpha$-strongly regular items]
\label{thm:intro:vcg:alpha-regular}
For any $\alpha \in (0, 1]$, $n\geq 1$ bidders and $m\geq 1$ asymmetric $\alpha$-strongly regular items, the tight competition complexity of the VCG auction against the welfare benchmark is $\constna$, i.e., $\vcg_{n+\constna} \geq \wel_n$. Here $\constna$ is a universal constant dependent only on $n$ and $\alpha$, satisfying
\begin{align*}
    \constna\in\left(\max\left\{\frac{1}{\alpha}-1,1\right\}\cdot n, \frac{11}{\alpha}\cdot n\right].
\end{align*}
Furthermore, having $\constna$ additional bidders are necessary even against the Bayesian optimal mechanism. Specifically, for every $\alpha \in (0, 1)$ and $n\in\naturals$, there exists an instance with {$m$} symmetric $\alpha$-strongly regular items such that {as the number of items, $m$, approaches infinity}, the expected revenue of the VCG auction for $n + \constna - 1$ is strictly less than the Bayesian optimal revenue, i.e., $\vcg_{n+\constna - 1} < \rev_n$.
\end{theorem}

In contrast to the tight competition complexity bound for regular distributions, which depends on both the number of bidders, $n$, and the number of items, $m$, the theorem above shows that for $\alpha$-strongly regular distributions with $\alpha > 0$, the tight competition complexity $\constna$ of the VCG auction exhibits \emph{no} dependence on the number of items, $m$, relying solely on a \emph{linear} dependence on the number of bidders, $n$, with the coefficient determined by $\alpha$. This observation aligns with our illustration in \Cref{example:intro:exp} and \Cref{fig:example-exp} for the exponential distribution, which belongs to the MHR class.

The formal definition of $\constna$ can be found in \Cref{def:alpha-regular:vcg:tight competition complexity}. At a high level, it is related to the order statistic of the generalized Pareto distribution \citep{Col-01}.\footnote{In the extreme cases of $\alpha = 0$ and $\alpha = 1$, the generalized Pareto distribution recovers the equal-revenue distribution and exponential distribution, respectively.}
Although the exact closed-form expression of the tight competition complexity bound $\constna$ is not provided, its lower and upper bounds are stated in \Cref{thm:intro:vcg:alpha-regular}. As a sanity check, when $\alpha$ approaches zero (corresponding to the class of regular distributions), both the lower and upper bounds diverge to infinity. Conversely, for $\alpha=1$ (corresponding to the class of MHR distributions), both the lower and upper bounds converge to $(e-1)n$ as number of bidders $n$ approaches infinity.

Since the Bayesian optimal mechanism and its expected revenue are typically complex to characterize, it is common in the literature to compare a given mechanism against a benchmark that provides an upper bound on the Bayesian optimal revenue. Following this approach, we analyze the competition complexity of the VCG auction relative to the welfare benchmark, also known as the first-best benchmark.
Surprisingly, while the competition complexity upper bound is derived with respect to the welfare benchmark, it turns out to be tight even when evaluated against the original Bayesian optimal revenue benchmark. Specifically, we can construct a problem instance where recruiting an additional $\constna - 1$ bidders and running the VCG auction is insufficient to suppress the revenue of the Bayesian optimal mechanism achieved with the original $n$ bidders.

At a high level, our upper-bound analysis (\Cref{lem:worst alpha}) demonstrates that the worst-case instance (with respect to the welfare benchmark) corresponds to the generalized Pareto distribution. Given that bidders have additive valuation functions, the VCG auction essentially reduces to running a second-price auction independently for each item. For any fixed item, its contribution to the welfare benchmark is the expected value of the first-order statistic,\footnote{Consistent with auction design literature, we refer to the $k$-th \emph{largest} (rather than the \emph{smallest}) value as the $k$-th order statistic.} while its contribution to the VCG auction is the expected value of the second-order statistic.
To establish that the generalized Pareto distribution is the worst case, we employ a multi-step argument. First, we identify a ``three-interval structure" (\Cref{lem:three intervals}) that facilitates a comparison between the first and second order statistics under varying sample sizes. Next, we leverage the tail bounds implied by $\alpha$-strong regularity to analyze these order statistics (\Cref{lem:alpha boundary}). Finally, we complement with the lower bound analysis by showing that for the generalized Pareto distribution, the Bayesian optimal revenue converges to the welfare benchmark as the number of items, $m$, approaches infinity (\Cref{lem:rev_is_wel}).

\xhdr{Competition complexity of the bundle-based second-price auction.} In the second part of this paper, we study on the bundle-based second-price auction and study its competition complexity for answering Question 2 imposed earlier in our introduction. (See more motivations behind studying the bundle-based second-price auction in \Cref{apx:bspa motivation}.)
We focus on the additional bidders required in the bundle-based second-price auction to surpass the Bayesian optimal revenue for a single bidder, which is already a challenging task.
Our results for the bundle-based second-price auction are threefold.\footnote{In \Cref{bspa_multiple_n}, we prove that for multi-bidder instances (i.e., $n\geq 2$), the competition complexity of the bundle-based second-price auction can be as bad as $\Omega(\exp(m))$, even if $m$ items have i.i.d.\ values from uniform distribution, which is MHR. An interesting avenue for future research is to design mechanisms---possibly combining the VCG auction and the bundle-based second-price auction---that achieve better competition complexity guarantees for multi-bidder instances.}

First, we continue our exploration about the competition complexity when the bidder has MHR valuation distributions for all items. (The formal theorem statement can be found in \Cref{thm:bspa mhr}.)
\begin{restatable}[Competition complexity of {\bspa} for MHR items]{theorem}{thmCompComplexityMHR}
\label{thm:intro:bspa:mhr}
For $n=1$ bidder and any $m\geq 1$ asymmetric MHR items, the competition complexity of the bundle-based second-price auction against the welfare benchmark is at most $3$, i.e., $\bspa_{4} \geq \wel_1$.
\end{restatable}
\Cref{thm:intro:bspa:mhr} aligns with the numerical example in \Cref{example:intro:exp} for the exponential distribution (a paradigmatic MHR distribution). This provides strong confidence that the bundle-based second-price auction is a natural and promising mechanism to study in the multi-item setting when the bidder has an additive valuation function. To establish this result, we employ a two-step argument. First, leveraging the technique developed for \Cref{thm:intro:vcg:alpha-regular}, we show that for selling a single MHR item, the competition complexity of the second-price auction against the welfare benchmark is at most 3. Second, we utilize the fact that the sum of independent (but not necessarily identical) random variables from MHR distributions remains MHR \citep{BMP-63}.

Next, we relax the MHR assumption and consider the case where item distributions are regular. Due to the significant challenges in the multi-item revenue maximization problem, many prior works \citep[e.g.,][]{MV-06,WT-14,GK-15,GK-18,BNR-18,TSN-19} have focused on the important special case of selling two or three items. As our second result for the bundle-based second-price auction, we study its competition complexity when there are at most three regular items. (The formal theorem statement can be found in \Cref{thm:bspa:3 items}.)
\begin{theorem}[Competition complexity of {\bspa} for regular items under small number of items]
\label{thm:intro:bspa:regular:3 items}
    For $n=1$ bidder and any $m\in\{2, 3\}$ asymmetric regular items, the competition complexity of the bundle-based second-price auction against the duality benchmark is at most $m$, i.e., $\bspa_{1 + m} \geq \cdw_1$.
\end{theorem}

We remark that \citet{BK-96} can be viewed as \Cref{thm:intro:bspa:regular:3 items} with $m = 1$ regular item {in the special case of $n = 1$}. Moreover, for $m \in\{2, 3\}$ regular items, our competition complexity bound for the bundle-based second-price auction is slightly better than the state-of-the-art bound for the VCG auction (which is the only mechanism with known competition complexity for our problem).  Specifically, \citet{BW-19} shows that the competition complexity bound for selling $m \in\{2, 3\}$ regular items with the VCG auction is at most 4.\footnote{\citet{BW-19} show that for $n = 1$ bidder and any $m\geq 1$ asymmetric regular items, the competition complexity of the VCG auction against the duality benchmark is at most $\lceil \ln(1 + m) + 2\rceil$, which is equal to 4 for both $m = 2$ and $m = 3$.}

To establish \Cref{thm:intro:bspa:regular:3 items}, we rely on a combinatorial argument that 
centers on comparing two quantities through a uniformly random $m \times (m+1)$ quantile matrix $\quantMatrix$. 
Concretely, we recast an upper bound of $\cdw_1$ and a lower bound of $\bspa_{m+1}$ 
as expectations of functions of $\quantMatrix$, which we call $\cdw_1(\quantMatrix)$ and $\bspa_{m+1}(\quantMatrix)$. 
We then show that for $m = 2$ and $m = 3$, $\cdw_1(\quantMatrix) \leq \bspa_{m+1}(\quantMatrix)$ holds for 
every possible choice of $\quantMatrix$. Along the way, we also derive revenue benchmarks 
for regular items, which may be useful in related problems.

Finally, we ask the following question motivated by the original result in \citet{BK-96}: by adding a single additional bidder, can the bundle-based second-price auction achieves good approximation against the Bayesian optimal mechanism?\footnote{Recall \citet{BK-96} show that when selling a single regular item, adding just one additional bidder in the second-price auction is sufficient to outperform the Bayesian optimal mechanism. While most subsequent works study competition complexity—quantifying the number of additional bidders required for prior-independent mechanisms to beat (or nearly match) the Bayesian optimal mechanism—there are also prior studies \citep[e.g.,][]{FLR-19,FJ-24} that examine approximation guarantees when recruiting only a single additional bidder.} We answer this question positively as follows. (The formal theorem statements can be found in \Cref{prop:bvcg constant approx}.)
\begin{theorem}[Constant approximation of {\bspa} for regular items]
\label{thm:intro:bspa:regular:adding one bidder}
    For $m \geq 1$ asymmetric regular items, the expected revenue of the bundle-based second-price auction for two bidders is a 48-approximation to the Bayesian optimal revenue, i.e., $48\cdot \bspa_2 \geq \rev_1$.
\end{theorem}
We remark that a similar statement fails for the VCG auction.
Specifically, it is known that for every $m\geq 1$, there exists an instance where all $m$ items are i.i.d.\ regular (in fact, realized from the equal revenue distribution) such that the VCG auction with two bidders only obtain $O(1/\log(m))$ fraction of the Bayesian optimal revenue for a single bidder \citep{FFR-18}.

We establish this 48-approximation for regular items (\Cref{thm:intro:bspa:regular:adding one bidder}) by analyzing the probability tail bounds induced by regularity and applying the standard core-tail decomposition technique \citep{BILW-14}. En route, our analysis framework also shows that for a regular bidder, selling all items as a grand bundle ({\brev}) is a 4-approximation to selling all items separately ({\srev}). Combining with the prior result that the better of selling grand bundle or selling items separately approximates the Bayesian optimal revenue \citep{BILW-14,CDW-16}, we prove that selling the grand bundle is an 18-approximation to the Bayesian optimal revenue (\Cref{thm:intro:brev}), which may be of independent interest. As an interesting comparison, it is known that without the regularity assumption, selling all items as a grand bundle does not guarantee a constant approximation against the Bayesian optimal revenue. In this sense, our result shows that regularity is a sufficient condition for achieving a constant approximation when selling all items as a grand bundle. (The formal theorem statements can be found in \Cref{cor:18brev_geq_rev}.)

\begin{theorem}[Constant approximation of {\brev} for regular items]
\label{thm:intro:brev}
    For $n = 1$ bidder and $m \geq 1$ asymmetric regular items, the optimal expected revenue from selling all items as a grand bundle is a 18-approximation to the Bayesian optimal revenue, i.e., $18\cdot \brev_1 \geq \rev_1$.
\end{theorem}

In addition, by combining results from \citet{BMP-63,BK-96,DRY-15}, we show that when items are MHR, the approximation guarantees for $\bspa_2$ and $\brev_1$ in \Cref{thm:intro:bspa:regular:adding one bidder,thm:intro:brev} improve to a factor of~$e$ even with respect to the welfare benchmark. This $e$-approximation also holds for the VCG auction with two bidders and for selling items separately to a single bidder. See \Cref{prop:brev approx wel mhr,prop:srev approx wel mhr,prop:bspa approx wel mhr,prop:vcg approx wel mhr}.

\subsection{Related Work} 
Our work connects to several strands of existing literature.

\xhdr{Competition complexity.} The most relevant to our work are results on competition complexity in auctions. \cite{BK-96} initiated the work in competition complexity for a single item. 
\cite{EFFTW-17, BW-19, DRWX-24}
studies the competition complexity
for additive bidders in multi-item settings, resolving the competition complexity of VCG auction for regular bidders ($\Omega(n \cdot \ln \frac{m}{n})$ when $m \gg n$, $\sqrt{mn}$ when $n \gg m$). Our results on VCG auction instead focus on $\alpha$-regular bidders, and obtain a competition complexity independent of $m$. \cite{FFR-18,CS-21} considers competition complexity for a broader group of auctions, but with a slightly different objective of recovering 
$(1 - \epsilon)$-fraction of the optimal revenue. In particular, \cite{FFR-18}
show that in the single-bidder setting, $\bspa$ has constant (dependent on $\epsilon$) competition complexity against $(1 - \epsilon)$ revenue benchmark. In contrast, our competition complexity results for $\bspa$ focus on recovering $100\%$ of the optimal revenue, which \cite{BW-19} notes requires substantially different technical tools. For MHR distributions, \citet{LP-18} establish a competition complexity of $[n, 3n]$ for VCG relative to welfare; we prove a similar result for $\bspa$.

% For MHR distributions, \citet{LP-18} show that for MHR distributions, the competition complexity of VCG, when benchmarked against welfare, lies within $[n, 3n]$. We prove an analogous result for $\bspa$. 
%%\cite{CS-21} shows that constant approximation ratio of a mechanism translate into constant (multiplicative factor) of enhanced competition, unless $\vcg$ with constant enhanced competition achieves close to optimal revenue. Utilizing this fact, they design simple auctions (that runs either $\vcg$, or a known constant competitive auction, each with constant probability) that with constant enhanced competition exceeds the original optimal revenue. 

\xhdr{Simple and approximately optimal mechanisms.} A broad line of work investigates simple and approximately optimal mechanisms in auction theory \citep[e.g.,][]{CHK-07, CHMS-10, CMS-15, HN-12, LY-13, BILW-14, BDHS-15, yao-15, CMS-15, CM-16, CDW-16, CZ-17, EFFTW-17, RW-18, HR-09, FLR-19}. Particularly relevant to our work is the approximation ratio of bundle-based mechanisms. \citet{HN-12} first establishes a $O(\log m)$ approximation guarantees for selling the grand bundle in the single buyer setting when the bidder's value for $m$ items are independent and identically distributed. {\citet{LY-13} improve the approximation guarantee of selling the grand bundle in \citet{HN-12} from $O(\log m)$ to constant. They also show an approximation lower bound of $\Omega(\log(m))$ for selling items separately, which illustrates a strict separation between selling items separately and selling the grand bundle in their setting.} More recently, \citet{AKLS-23} analyze the competitiveness of bundling in buy-many settings for an additive buyer and arbitrarily correlated value distribution. Our constant approximation result of $\bspa$ in \Cref{sec:approx_bundle} focus on a setting similar to that of \citet{HN-12,LY-13}, but for non-identical items. 

%\smwcomment{In the above list, might make sense to note that LiYao13 improved the logm approximation for iid items to a constant?}

\xhdr{Optimality of bundle selling.}
In the multi-item settings, the selling items as a grand bundle has been proven to be optimal under certain sufficient conditions \citep{MV-06,DDT-17,GK-18,HH-21,yan-22,Ghi-23}. Unlike our work which assumes the value for each item is independent and belong a common distribution family (e.g., regular, strongly regular, MHR) , all sufficient conditions require either specific correlation or focus on specific distributions (e.g., uniform). In the same direction, the sufficient conditions for the optimality of nested bundle have been studied in \citet{BBHS-21,yan-22}.

\section{Preliminaries}
\label{sec:prelim}

We consider the setting where there are $m$ items and $n$ i.i.d.\ bidders. Bidders have additive valuation functions: given value profile $\{\val_{j}\}_{j\in[m]}$ over $m$ items, a bidder's valuation for any subset $S$ of items is $\sum_{j\in S}\val_{j}$.
The value profile of each bidder is drawn from product distribution $\dist = \Pi_{j\in[m]} \dist_j$, where $\dist_j$ is the marginal value distribution for item $j$.\footnote{Unless stated otherwise, we assume that each bidder’s value for each item is realized independently.}
By slightly abusing the notations, we use $\cdf_j(\cdot)$ to denote the cumulative density function (CDF).
We assume that each distribution $\dist_j$ is continuous and atomless, so $\cdf_j$ is a strictly increasing function. Let $\pdf_j(\cdot)$ denote the corresponding probability density function (PDF). 
For any quantile $\quant \in [0,1]$, we define $\cdf_j^{-1}(\quant)$ as the value in distribution~$\dist_j$ corresponding to quantile $\quant$. 

We introduce the following notations for the welfare and revenue benchmarks in this paper, each in the context of $n$ bidders with values drawn i.i.d.\ from distribution $\dist$. 
\begin{itemize}
    \item $\rev_n(\dist)$: the expected revenue of the Bayesian revenue-optimal mechanism. 
    \item $\bspa_n(\dist)$: the expected revenue generated by bundling all items together, and running the second-price auction. 
    \item $\brev_n(\dist)$: the expected revenue generated by bundling all items together, and running Myerson's revenue-optimal single-item auction \citep{mye-81}.
    \item $\vcg_n(\dist)$: the expected revenue of the VCG auction that allocates the items efficiently, which, in the context of additive bidders, is equivalent to selling each item separately through the second-price auction. 
   \item $\srev_n(\dist)$: the expected revenue generated by running Myerson's revenue-optimal single-item auction for each item separately.
    \item $\wel_n(\dist)$: the optimal expected welfare, which can be achieved by the VCG auction.
    
    \item $\cdw_n(\dist)$: the expected value of the duality benchmark. See \Cref{sec:prelim:duality} for the formal definition. 
\end{itemize}
When it is clear from the context, we omit $\dist$ from above notations.
All five mechanisms and two benchmarks form the hierarchy 
\begin{align*}
    \bspa\leq\brev\leq \rev\leq \cdw \leq \wel
    \\
    \vcg\leq \srev\leq \rev\leq\cdw\leq \wel
\end{align*}
with incomparable pairs ``{\vcg} vs.\ {\bspa}'' and ``{\srev}'' vs.\ ``{\brev}''.

\subsection{Distributional Assumptions}
\label{sub:dist_assumption}
Our paper considers distributions under the following assumptions, listed in increasing order of restrictiveness: regular, $\alpha$-strongly regular, and monotone hazard rate (MHR).\footnote{Here we introduce the distributional assumptions for the univariate distribution $\dist$. For (joint) product distributions, we say it is regular/MHR/$\alpha$-strongly regular if all its marginal distributions satisfy the corresponding distributional assumption.} Regular and MHR distributions are standard assumptions extensively used in the auction design literature. The $\alpha$-strongly regular distribution, introduced by \cite{CR-14}, serves as a natural interpolation between the regular and MHR assumptions.

Given any (univariate) distribution $\dist$, let 
\begin{align*}
\hazard(\val) \triangleq \frac{\pdf(\val)}{1 - \cdf(\val)}
\end{align*}
be the hazard rate for distribution $\dist$, and let
\begin{align*}
    \virtual(\val) \triangleq \val - \frac{1}{\hazard(\val)}
\end{align*}
be the corresponding virtual value function.

\begin{definition}[Regular distribution]
A distribution $\dist$ is \emph{regular} if its virtual value function $\virtual(\val)$ is non-decreasing in $\val$.
\end{definition}

\begin{definition}[$\alpha$-strongly regular distribution]
Given any parameter $\alpha\in[0,1]$, a distribution $ \dist $ is $\alpha$-\emph{strongly regular} if its virtual value function $\virtual(\val)$
is $\alpha$-strongly monotone, i.e., $\virtual(\val) -\virtual(\val\primed) \geq \alpha (\val - \val\primed)$ for all $\val,\val\primed$ in the support of $\dist$.\footnote{When the virtual value function $\virtual(\cdot)$ is continuous, an equivalent statement is that derivative of the virtual value function $\virtual'(\val) \geq \alpha$.}
\end{definition}

\begin{definition}[Monotone hazard rate (MHR) distribution]
A distribution $ \dist $ is \emph{MHR} if its hazard rate $\hazard(\val)$ is non-decreasing in $\val$.
\end{definition}

Note that by definition, a distribution $\dist$ is regular if and only if it is $0$-strongly regular, and is MHR if and only if it is $1$-strongly regular.

\subsection{Enhanced Competition and Duality Benchmark} 
\label{sec:prelim:duality}
In this section, we formally define competition complexity and the duality benchmark {\cdw}, which will be useful for bounding competition complexity in combinatorial auctions.

\begin{definition}[Competition complexity]
In an auction setting with $m$ items and $n$ bidders whose values for the items are drawn i.i.d.\ from distribution $\dist = \Pi_{j\in[m]} \dist_j$, the \emph{competition complexity} of an auction $\auc$ is the minimum number of additional bidders $c$ such that 
\begin{align*}
\auc_{n+c}(\dist) \geq \rev_n(\dist)   
\end{align*}
where $\auc_{n+c}(\dist)$ and $\rev_n(\dist)$ are the expected revenue from auction $\auc$ with $n + c$ bidders and the Bayesian revenue-optimal mechanism with $n$ bidders, respectively.
\end{definition} 

In the special case of single-item auctions, \citet{BK-96} show that the competition complexity of the VCG auction for regular buyers is $1$. 
\begin{theorem}[\citealp{BK-96}] \label{thm:BK}
For $n\geq 1$ bidder and $m = 1$ regular item, the tight competition complexity of the second-price auction against the revenue benchmark is 1, i.e., $\vcg_{n+1} \geq \rev_n$. 
\end{theorem}

Next we introduce the formal definition of the duality benchmark.

\begin{theorem}[Quantile-based duality benchmark, \citealp{CDW-16, EFFTW-17}] \label{thm:duality}
   Consider an auction setting with $m$ items and $n$ regular bidders whose valuations are drawn i.i.d.\ from distribution $\dist$. Let $Q$ be an $n \times m$ matrix, where each entry $\quant_{ij}$ is drawn i.i.d.\ from uniform distribution $U[0,1]$, and for each row $i$, let $\quant^*_{i1}, \quant^*_{i2}, \cdots, \quant^*_{im}$ be $\quant_{i1}, \cdots, \quant_{im}$ sorted in decreasing order (with ties broken randomly). Then, 
   \begin{align*}
        \rev_n(\dist) \leq \expect[Q \sim U{[0,1]}^{ n \times m}]{\sum_{j\in[m]} \max_{i \in [n]} \left\{
            \plus{\virtual_j\left(\cdf_j^{-1}(\quant_{ij})\right)} \cdot \indicator{\quant_{ij} = \quant^*_{i1}} +  \cdf_j^{-1}(\quant_{ij}) \cdot \indicator{\quant_{ij} \neq \quant^*_{i1}}
            \right\}
        }
   \end{align*}
   where $\virtual_j$ is the virtual value function for item $j$ and operator $\plus{\cdot}\triangleq \max\{\cdot, 0\}$. We refer to the right-hand side as the \emph{(quantile-based) duality benchmark}, denoted $\cdw_n(\dist)$.
\end{theorem}
\begin{corollary} \label{cor:duality_one_buyer}
In the special case of $n = 1$, letting $\vec{\quant}=(\quant_1, \cdots \quant_m)$ be quantiles drawn i.i.d.\ from uniform distribution $U[0,1]$ and $\sigma_{\vec{\quant}}: [m] \rightarrow [m]$ be the permutation that sorts $\quant_1, \cdots \quant_m$ in descending order (i.e. $\sigma_{\vec{\quant}}(j)$ is the $j^{th}$ largest order statistics among $\vec{\quant}$), the duality benchmark has value 
\begin{align*}
    \cdw_1(\dist) = \expect[\vec{\quant} \sim U{[0,1]}^{m}] 
    {
    \plus{\virtual_{\sigma_{\vec{\quant}}(1)}\left(\cdf_{\sigma_{\vec{\quant}}(1)}^{-1}\left(\quant_{\sigma_{\vec{\quant}}(1)}\right)\right)}  +  \sum_{j\in[2:m]} \cdf_{\sigma_{\vec{\quant}}(j)}^{-1}(\quant_{\sigma_{\vec{\quant}}(j)}) 
    }
\end{align*}
\end{corollary}

\section{Competition Complexity for Welfare}
\label{sec:wel}

In this section, we first characterize the competition complexity of the VCG auction relative to the welfare benchmark and then show that this bound remains tight even when compared to the optimal revenue benchmark. 
We further establish that the competition complexity of selling items as a grand bundle can be unbounded when there are at least two bidders. 
For upper bounds on competition complexity, we require only that the distributions are i.i.d.~across bidders, allowing for arbitrary heterogeneity and correlation in values across items. In contrast, our lower-bound constructions restrict attention to distributions that are independent across items. Our results rely on a novel bound on the order statistics for $\alpha$-strongly regular distributions.

\begin{definition}[Competition complexity constant of {\vcg} for $\alpha$-strongly regular items]
\label{def:alpha-regular:vcg:tight competition complexity}
For any $\alpha\in(0,1)$ and $n \in \naturals$, define constant $\constna \in \naturals$ as the minimum integer such that
\begin{align*}
    \dist_{2:n+\constna} = \alpha \dist_{1:n+\constna} - (1-\alpha) \geq \dist_{1:n}
\end{align*}
where $\dist$ is the generalized Pareto distribution with CDF $\cdf(v) = 1-(1+\val)^{-\frac{1}{1-\alpha}}$, and $\dist_{k:n}$ is its expected $k$-th highest value given $n$ samples.

For $\alpha=1$, let $\constna \in \naturals$ be the minimum integer such that 
\begin{align*}
    \dist_{2:n+\constna} = \dist_{1:n+\constna} - 1 \geq \dist_{1:n}
\end{align*}
where $\dist$ is the exponential distribution with CDF $\cdf(\val) = 1-e^{-\val}$.
\end{definition}

\noindent\textbf{Remark:} The equality for deriving $\dist_{2:n+\constna}$ in the above definition can be directly computed based on the functional form of the distribution.
The constant $\constna$ for $\alpha=1$ can be easily computed due to the special structure of the exponential distribution. Specifically, through a simple algebraic calculation, 
for the exponential distribution $\dist$, we have 
\begin{align*}
\dist_{2:n+\const_{n,1}} = H_{n+\const_{n,1}}
\quad\text{and}\quad
\dist_{1:n} = H_n
\end{align*}
where $H_n = \sum\nolimits_{i\in[n]}\frac{1}{i}$ is the harmonic number.
Therefore, $\const_{n,1}$ is the minimum integer such that 
$H_{n+\const_{n,1}} - H_n \geq 1$.
Note that $\const_{n,1}\to (e-1)\cdot n$ when $n\to\infty$. 

When $\alpha\in (0,1)$, the dependence of $\constna$ on $n$ is less transparent in the definition. We provide an asymptotic bound on $\constna$ for $\alpha\in (0,1)$ in the following lemma. 

\begin{lemma}\label{lem:bound_of_C}
For any $\alpha\in(0,1]$ and any $n\geq 1$, the competition complexity constant $\constna$ satisfies
\begin{align*}
    \constna\in\left(\max\left\{\frac{1}{\alpha}-1,1\right\}\cdot n, \frac{11}{\alpha}\cdot n\right]
\end{align*}
\end{lemma}
\begin{proof}
First, observe that for any $\epsilon\in [0,1]$ and any $\alpha\in (0,1)$, we have 
\begin{align}\label{eq:eps_alpha}
1-\epsilon \geq \alpha(1-\epsilon^{\frac{1}{\alpha}}).
\end{align}
The inequality holds since $1-\epsilon - \alpha(1-\epsilon^{\frac{1}{\alpha}})$ is monotonically decreasing in $\epsilon$, and the value of this difference is $0$ when $\epsilon=1$. 
Therefore, letting $N=\frac{n}{\alpha}$, we have 
\begin{align*}
\alpha \dist_{1:N} - (1-\alpha) < \alpha \dist_{1:N}
&= \int_0^{\infty} \alpha\cdot \rbr{1-(1-(1+\val)^{-\frac{1}{1-\alpha}})^N} \dd v \\
&\leq \int_0^{\infty} 1-(1-(1+\val)^{-\frac{1}{1-\alpha}})^n \dd v
= \dist_{1:n}.
\end{align*}
The second inequality holds by applying \Cref{eq:eps_alpha} with $\epsilon = (1-(1+\val)^{-\frac{1}{1-\alpha}})^n$.
This implies that $\constna > (\frac{1}{\alpha}-1)\cdot n$. 

Moreover, note that for any $\alpha\in (0,1)$, the exponential distribution is also a special case of $\alpha$-strongly regular distribution. 
Note that $\const_{n,1}$ is the minimum integer such that 
$H_{n+\const_{n,1}} - H_n \geq 1$,
where $H_n = \sum\nolimits_{i\in[n]}\frac{1}{i}$ is the harmonic number.
This implies that for the exponential distribution, $\const_{n,1} \geq n$ for any $n$, 
which further provides a lower bound on $\constna$ for any $\alpha\in(0,1)$ and $n\geq 1$.

Next, we will provide the upper bound for $\constna$. 
Note that letting $\tau=n^{1-\alpha}-1$, 
we have $(1+\tau)^{-\frac{1}{1-\alpha}} = \frac{1}{n}$
and 
\begin{align}
\begin{split}
\dist_{1:n} 
&= \int_0^{\infty} 1-\left(1-(1+\val)^{-\frac{1}{1-\alpha}}\right)^n \dd v \\
&\leq \tau + \int_{\tau}^{\infty} n\cdot (1+\val)^{-\frac{1}{1-\alpha}} \dd v 
= n^{1-\alpha}\cdot \left(1+\frac{1-\alpha}{\alpha}\right) - 1.\label{eq:upper_f1}
\end{split}
\end{align}
The inequality holds since $(1-\epsilon)^n \geq \max \{1-n\epsilon, 0\}$ for any $\epsilon\in [0,1]$.
In addition, letting $\hat{\tau}=(2n)^{1-\alpha}-1$, we have
\begin{align}
\begin{split}
\dist_{1:n} 
&\geq \int_{\hat{\tau}}^{\infty} 1-\left(1-(1+\val)^{-\frac{1}{1-\alpha}}\right)^n \dd v \\
&\geq \int_{\hat{\tau}}^{\infty} \frac{n}{2}\cdot (1+\val)^{-\frac{1}{1-\alpha}} \dd v 
= (2n)^{1-\alpha}\cdot \frac{1-\alpha}{2\alpha}.\label{eq:lower_f1}
\end{split}
\end{align}
where the second inequality holds since $(1-\epsilon)^n \leq 1-\frac{n\epsilon}{2}$ for any $\epsilon \in [0,\frac{1}{2n}]$.
Thus for $N\geq n\cdot (\frac{2^\alpha}{\alpha(1-\alpha)})^{\frac{1}{1-\alpha}}$, 
by combining \Cref{eq:upper_f1,eq:lower_f1}, we have 
\begin{align*}
\alpha \dist_{1:N} - (1-\alpha)
\geq \alpha (2N)^{1-\alpha}\cdot \frac{1-\alpha}{2\alpha} - 1
\geq n^{1-\alpha}\cdot \frac{1}{\alpha} -1
\geq \dist_{1:n}.  
\end{align*}
Note that for $\alpha \leq \frac{1}{5}$, we have $(\frac{2^\alpha}{\alpha(1-\alpha)})^{\frac{1}{1-\alpha}} \leq \frac{2.351}{\alpha}$. 
Therefore, we have $n+\constna \leq \frac{2.351n}{\alpha}$ for $\alpha \leq \frac{1}{5}$. 
For $\alpha \geq \frac{1}{5}$, since $\alpha$-strong regularity is a special case of $\frac{1}{5}$-regularity, we have $\constna \leq 11n \leq \frac{11n}{\alpha}$.
\end{proof}

\begin{theorem}[Competition complexity of {\vcg} against {\wel} for $\alpha$-strongly regular items]
\label{thm:welfare}
For any $\alpha \in (0, 1]$, $n\geq 1$ bidders and $m\geq 1$ (possibly correlated) asymmetric $\alpha$-strongly regular items, the tight competition complexity of the VCG auction against the welfare benchmark is $\constna$, i.e., $\vcg_{n+\constna} \geq \wel_n$.
Moreover, there exists an instance such that $\vcg_{n+\constna-1} < \wel_n$.
\end{theorem}
Note that our theorem applies to distributions with arbitrary correlations instead of just product distributions. 
We prove our theorem by identifying the exact worst case $\alpha$-strongly regular distributions given any parameter $\alpha\in (0,1)$ are generalized Pareto distributions. 
The procedure for showing that the worst-case distribution for $\alpha=1$ is the exponential distribution follows the same steps as before and is therefore omitted in this paper.
\begin{lemma}\label{lem:worst alpha}
For any $\alpha\in(0,1)$, $N, n\in\naturals$ with $N> n \geq 1$,
and $\alpha$-strongly regular distribution $\dist$, 
let $\hat{\dist}$ be the generalized Pareto distribution such that 
$1-\hat{\cdf}(\val) = (1+\val)^{-\frac{1}{1-\alpha}}$.
If $\hat{\dist}_{2:N} \geq \hat{\dist}_{1:n}$, 
then $\dist_{2:N} \geq \dist_{1:n}$.
\end{lemma}
The proof of \cref{lem:worst alpha} relies on the following two technical lemmas.
\begin{lemma}[Three-interval structure]
\label{lem:three intervals}
For any distribution $\dist$, and $N, n \in\naturals$ with $N > n \geq 1$,
there exists 
$0\leq \quant\primed \leq \quant\doubleprimed \leq 1$
such that 
\begin{align*}
    &\contribution_{2:N}(\quant)\leq \contribution_{1:n}(\quant)
    &
    \quant \in [0,\quant\primed]\cup[\quant\doubleprimed, 1]
    \\
    &
    \contribution_{2:N}(\quant)\geq \contribution_{1:n}(\quant)
    &
    \quant\in[\quant\primed,\quant\doubleprimed]
\end{align*}
where $\contribution_{k:n}(\quant)$ is defined as the probability density that 
the value with quantile~$\quant$ (i.e., $\cdf^{-1}(\quant)$)
is the $k$-th highest order statistic among~$n$ i.i.d.\ samples from distribution $\dist$. 
\end{lemma}
\begin{proof}
By definition, 
\begin{align*}
   \contribution_{1:n}(\quant) = n\quant^{n-1} 
   \;\;\mbox{and} \;\;
    \contribution_{2:N}(\quant) = 
    N(N-1)(1-\quant) \quant^{N-2}.
\end{align*}
Thus, 
\begin{align*}
\contribution_{1:n}(\quant) - 
    \contribution_{2:N}(\quant)
    =
    \quant^{n-1}
    \left(
    n - N(N-1)(1-\quant)\quant^{N - n - 1}
    \right).    
\end{align*}
Since $(1-\quant)\quant^{N-n-1}$
is a single-peaked function 
(i.e., first increasing and then decreasing),
we finish the proof as desired. 
\end{proof}

\begin{lemma}[\citealp{ABB-22}]
\label{lem:alpha boundary}
For any $\alpha\in [0,1]$, $\val_1,\val_2\in\reals_+$ with $\val_1 \geq \val_2$,
$\quant\primed,\quant\doubleprimed\in[0,1]$ with 
$\quant\primed > \quant\doubleprimed$, and
$\alpha$-strongly regular distribution $\dist$ with CDF satisfying that 
$\cdf(\val_1) = \quant\primed, \cdf(\val_2) = \quant\doubleprimed$, we have 
\begin{align*}
1 - \cdf(\val) &\geq \quant\doubleprimed \cdot\Gamma_{\alpha}\left(\Gamma_{\alpha}^{-1}\left(\frac{\quant\primed}{\quant\doubleprimed}\right)\cdot \frac{\val-\val_2}{\val_1-\val_2}\right) & 
\val\in[\val_2,\val_1]\\
1 - \cdf(\val) &\leq \quant\doubleprimed \cdot\Gamma_{\alpha}\left(\Gamma_{\alpha}^{-1}\left(\frac{\quant\primed}{\quant\doubleprimed}\right)\cdot \frac{\val-\val_2}{\val_1-\val_2}\right) & \val\in[\val_1, \infty)
\end{align*}
where 
\begin{align*}
\Gamma_{\alpha}(v) = \begin{cases}
(1+(1-\alpha)v)^{-\frac{1}{1-\alpha}}& \alpha\in [0,1)\\
e^{-v} & \alpha = 1.
\end{cases}
\end{align*}
\end{lemma}

Below we formally prove \Cref{lem:worst alpha} and then \Cref{thm:welfare}.

\begin{proof}[Proof of \cref{lem:worst alpha}]
Note that for any distribution $\dist$, we have that 
% $\dist_{2:N} = \int_0^1 \contribution_{2:N}(\quant) \cdot v_{\dist}(\quant) \dd \quant$ and $\dist_{1:n} = \int_0^1 \contribution_{1:n}(\quant) \cdot v_{\dist}(\quant) \dd \quant$.
\begin{align*}
\dist_{2:N} = \int_0^1 \contribution_{2:N}(\quant) \cdot \cdf^{-1}(\quant) \dd \quant
\text{ and } \dist_{1:n} = \int_0^1 \contribution_{1:n}(\quant) \cdot\cdf^{-1}(\quant)\dd \quant.
\end{align*}
where $\contribution_{k:n}(\quant)$ is defined in \Cref{lem:three intervals}.
In addition, \cref{lem:three intervals} guarantees that
there exists $0\leq \quant\primed < \quant\doubleprimed\leq 1$ 
and corresponding value $\val_1>\val_2$ such that 
$\contribution_{2:N}(\quant) \leq \contribution_{1:n}(\quant)$ 
for $\quant\in [0,\quant\primed]\cup[\quant\doubleprimed,1]$
and $\contribution_{2:N}(\quant) \geq \contribution_{1:n}(\quant)$  for 
$\quant\in [\quant\primed,\quant\doubleprimed]$.
To minimize the difference $\dist_{2:N} - \dist_{1:n}$, 
we need to minimize $\cdf^{-1}(\quant)$ for 
$\quant\in [\quant\primed,\quant\doubleprimed]$
and maximize $\cdf^{-1}(\quant)$ for
$\quant\in [0,\quant\primed]\cup[\quant\doubleprimed,1]$.
By \cref{lem:alpha boundary}, 
the extreme $\alpha$-strongly regular distribution $\tilde{\dist}$
that crosses $(\val_1,\quant\primed)$ and $(\val_2,\quant\doubleprimed)$ satisfies 
\begin{align*}
1- {\tilde{\dist}}(\val) 
&= \quant\doubleprimed \cdot\Gamma_{\alpha}\left(\Gamma_{\alpha}^{-1}\left(\frac{\quant\primed}{\quant\doubleprimed}\right)\cdot \frac{\val-\val_2}{\val_1-\val_2}\right) \\
&= \quant\doubleprimed \cdot (1+\beta (\val-\val_2))^{-\frac{1}{1-\alpha}}
\end{align*}
where $\beta = \Gamma_{\alpha}^{-1}\left(\frac{1-\quant\primed}{1-\quant\doubleprimed}\right)\cdot \frac{1-\alpha}{\val_1-\val_2}$.
Let $\val_0$ be the constant such that 
$(1-\quant\doubleprimed) \cdot (1+\beta (\val_0-\val_2))^{-\frac{1}{1-\alpha}} = 1$.
By shifting the distribution $\tilde{\dist}$ down by $\val_0$, 
and scaling it by a factor $(1-\quant\doubleprimed)^{\frac{1}{1-\alpha}}\cdot \beta$, 
the resulting distribution is $\hat{\dist}$. 
Moreover, it is easy to verify that $\hat{\dist}_{2:N} \geq \hat{\dist}_{1:n}$
if and only if $\tilde{\dist}_{2:N} \geq \tilde{\dist}_{1:n}$, 
where the latter implies that $\dist_{2:N} \geq \dist_{1:n}$ by construction. 
\end{proof}

\begin{proof}[Proof of \cref{thm:welfare}]
By \cref{lem:worst alpha}, for $\alpha\in(0,1)$, it is without loss to consider $\alpha$-strongly regular distribution $\hat{\dist}$ with CDF 
$1-\hat{\dist}(\val) = (1+\val)^{-\frac{1}{1-\alpha}}$, 
where the corresponding virtual value is $\alpha \val - (1-\alpha)$.
Since the expected revenue from the second-price auction is the expected virtual value, 
we have $\hat{\dist}_{2:N} = \alpha \hat{\dist}_{1:N} - (1-\alpha)$.
By the definition of constant $\constna$ (\Cref{def:alpha-regular:vcg:tight competition complexity}), we have $\hat{\dist}_{2:\constna} \geq \hat{\dist}_{1:n}$. 
The same argument applies for MHR distributions when $\alpha=1$.
\end{proof}

\subsection{Tightness against The Bayesian Optimal Revenue}
In this section, we show that when the benchmark is switched from welfare to the optimal revenue, the competition complexity bound derived in \Cref{thm:welfare} remains necessary, even if values are i.i.d.~across different items.

\begin{restatable}[Competition complexity lower bound of {\vcg} against {\rev} for $\alpha$-strongly regular items]{theorem}{thmTightnessAgainstRevenue}
\label{thm:tightness_against_revenue}
For any $\alpha \in (0, 1]$, $n\geq 1$ bidders and $m\geq 1$ i.i.d.~$\alpha$-strongly regular items, the competition complexity of the VCG auction against the revenue benchmark is at least $\constna$, i.e., there exists an instance such that $\vcg_{n+\constna-1} < \rev_n$.
% For any $\alpha \in (0, 1]$ and $n\geq 1$, there exists a distribution $\dist$ that is i.i.d.~across bidders and items where the marginal distributions are $\alpha$-strongly regular
% such that the competition complexity of $\srev$
% with respect to the optimal revenue
% is at least $c(n;\alpha)$ for sufficiently large~$m$.
% That is, there exists $M\geq 1$ such that for any $m\geq M$, $\srev_{c(n;\alpha)-1}(\dist) < R_n(\dist)$.
\end{restatable}
We will first show that when the number of items goes to infinity, the optimal revenue actually converges to the welfare (see \cref{lem:rev_is_wel}). 
In this case, if the additional number of buyers is less than $\constna$, proof of \cref{thm:welfare} implies that for i.i.d.~generalized Pareto distributions, the revenue from each item is strictly less than the average welfare. 
Combining the observations, with less than $\constna$ additional buyers, the expected revenue from selling separately is strictly less than the optimal revenue when the number of items is sufficiently large. 
% The formal proof is provided in \cref{apx:proof_wel}.
% \yfcomment{Yingkai, could you include a formal proof for Theorem 3.6 in the appendix?}

\begin{proof}[Proof of \cref{thm:tightness_against_revenue}]
Given any $\alpha\in(0,1)$, let $\dist^{\otimes m}$ be the product distribution with marginal distribution~$\dist$ on each of $m$ coordinates where $\dist$ is the generalized Pareto distribution with CDF $\cdf(v) = 1-(1+\val)^{-\frac{1}{1-\alpha}}$.
Let~$\dist$ be the exponential distribution if $\alpha=1$. 
By \cref{thm:welfare}, for any $\alpha\in(0,1]$, there exists $\epsilon > 0$ such that for any $m\geq 1$, we have 
\begin{align*}
\frac{1}{m}\cdot\vcg_{n+\constna-1} < \frac{1}{m}\cdot\wel_n - \epsilon.
\end{align*}
Moreover, \cref{lem:rev_is_wel} shows that 
\begin{align*}
\lim_{m\rightarrow \infty} \frac{1}{m}\rev_n = 
    \dist_{1:n} = \frac{1}{m}\wel_n,
\end{align*}
which further implies that for any $\hat{\epsilon}>0$, there exists $\hat{m}\geq 1$ such that
\begin{align*}
\frac{1}{\hat{m}}\rev_n > 
\frac{1}{\hat{m}}\wel_n - \hat{\epsilon}. 
\end{align*}
Setting $\hat{\epsilon} = \frac{\epsilon}{2}$, 
we have 
\begin{align*}
\frac{1}{\hat{m}}\cdot\vcg_{n+\constna-1} < \frac{1}{\hat{m}}\cdot\rev_n - \frac{\epsilon}{2},
\end{align*}
which implies the theorem.
\end{proof}

\begin{lemma}\label{lem:rev_is_wel}
For any $\alpha\in(0, 1)$ and $n$ bidders with 
$\alpha$-strongly regular distribution $\dist$ such that 
$1-\cdf(\val) = (1 + \val)^{-\frac{1}{1-\alpha}}$, or $1-\cdf(\val) = e^{-\val}$ for $\alpha=1$, 
$\lim_{m\rightarrow \infty} \frac{1}{m}\rev_n(\dist^{\otimes m}) = 
\dist_{1:n}$.
\end{lemma}
\begin{proof}
Note that the mean of distribution $F$ is $\frac{1}{\alpha}-1$ or $\alpha\in (0,1)$ and the mean is $1$ for $\alpha=1$.
Note that the mean is always finite. 
Let $z_n=\dist_{1:n}-\dist_{2:n}$. 
For any constant $\epsilon>0$, 
consider the two-part tariff mechanism where each bidder is charged an entry fee of 
\begin{align*}
E=m\rbr{\frac{z_n}{n}-\epsilon}.
\end{align*}
If the bidder pays the entry fee, the bidder can purchase any item at the second highest price. 

Note that due to the i.i.d.\ assumption, the expected utility of each bidder $i$ on item $j$ without paying the entry fee is 
\begin{align*}
\expect{\max\cbr{0,\val_{ij}-\max_{i'\neq i}\val_{i'j}}} 
= \frac{z_n}{n}.
\end{align*}
By the weak law of large numbers, for any $\epsilon>0$, we have
\begin{align}\label{eq:wlln}
\lim_{m\rightarrow \infty} 
\prob{\abs{\frac{1}{m}\sum\nolimits_{j\in[m]}\max\cbr{0,\val_{ij}-\max_{i'\neq i}\val_{i'j}} - \frac{z_n}{n}} \geq \epsilon} 
= 0. 
\end{align}
Let $\event_1$ be the event that 
\begin{align*}
\abs{\frac{1}{m}\sum\nolimits_{j\in[m]} \max\cbr{0,\val_{ij}-\max_{i'\neq i}\val_{i'j}} - \frac{z_n}{n}} < \epsilon
\end{align*}
for all bidders $i$. 
By the union bound and \Cref{eq:wlln}, for any $\delta>0$, there exists $M_1>0$ such that for any $m\geq M_1$, $\prob{\event_1} \geq 1- \delta$. 

Similarly, let $\event_2$ be the event that 
\begin{align*}
\abs{\frac{1}{m}\sum\nolimits_{j\in[m]} \max_{i\in[n]}\val_{ij} - \dist_{1:n}} < \epsilon.
\end{align*}
Again by the weak law of large numbers, for any $\delta>0$, there exists $M_2>0$ such that for any $m\geq M_2$, $\prob{\event_2} \geq 1- \delta$. 
By the union bound, for any $m\geq \max\cbr{M_1,M_2}$, $\prob{\event_1\cap\event_2}\geq 1-2\delta$.

Conditional on event $\event_1$, all bidders will pay the entry fee and hence the allocation is efficient. Moreover, the utility of each bidder is at most $m\epsilon$. 
Conditional on event $\event_2$, the welfare from efficient allocation is at least $m(F_{1:n}-\epsilon)$. 
Therefore, the expected revenue is 
\begin{align*}
\rev_n(\dist^{\otimes m})
&\geq \prob{\event_1\cap\event_2} \cdot \rev_n(\dist^{\otimes m}\given \event_1\cap\event_2)\\
& \geq (1-2\delta) \cdot \rbr{m(F_{1:n}-\epsilon) - nm\epsilon}.
\end{align*}
Therefore, for any $m\geq \max\cbr{M_1,M_2}$, we have 
\begin{align*}
\frac{1}{m}\rev_n( \dist^{\otimes m}) 
\geq (1-2\delta) \cdot \rbr{F_{1:n} - (n+1)\epsilon}.
\end{align*}
Since the inequality holds for any $\delta>0$ and $\epsilon>0$, by taking $\delta,\epsilon\to 0$ and $m\to \infty$, 
we have 
\begin{align*}
    \lim_{m\rightarrow \infty} \frac{1}{m}\rev_n(\dist^{\otimes m}) = 
    \dist_{1:n}. 
\end{align*}
which finishes the proof as desired.
\end{proof}

\subsection{(In)approximation of Pure Bundling}
\label{sub:inapproximation_bundling}

In this section, we show that, in contrast to the VCG auction, which has a competition complexity $\Theta(n)$ for any $n\geq 1$ regardless of the number of items, the competition complexity of selling all items as a grand bundle is $\Omega(\exp(m))$ for any $n\geq 2$, even for uniform distributions.\footnote{In fact, it is easy to show that, using the same steps, this competition complexity lower bound extends to arbitrary non-degenerate distributions with finite moments of all orders.}
Finally, we complement this impossibility result by showing that when $n=1$, the competition complexity of selling all items as a grand bundle is constant for independent MHR distributions with respect to the welfare benchmark. 
% Omitted Proofs in this section can be found in \Cref{app:sec:cc_bundle}. 

\begin{restatable}[Competition complexity lower bound of {\bspa} against {\rev} for MHR items and multiple bidders]{theorem}{bspaMultipleN}
\label{bspa_multiple_n}
For $n\geq 2$ bidders and $m\geq 1$ i.i.d.\ items with values drawn from uniform distribution $U[0,1]$, the competition complexity against the revenue benchmark of $\brev$ or $\bspa$ is $\Omega(\exp(m))$, i.e., $\bspa_{n + o(\exp(m))} < \OPT_n$ and $\brev_{n + o(\exp(m))} < \OPT_n$.
\end{restatable}
\begin{proof}
First note that \cref{lem:rev_is_wel} implies that the optimal revenue converges to the optimal welfare for $\alpha$-strongly regular distributions for any $\alpha\in(0,1)$. 
Uniform distributions are $\alpha$-strongly regular and hence the optimal revenue in this case converges to the optimal welfare when $m$ is sufficiently large. 
That is, 
\begin{align*}
    \lim_{m\rightarrow \infty} \frac{1}{m}\rev_n(\dist^{\otimes m}) = 
    \dist_{1:n}. 
\end{align*}
For any bidder $i\in[n]$, by Hoeffding's inequality, we have 
\begin{align*}
\prob{\sum\nolimits_{j\in[m]}\val_{ij} \geq \rbr{\dist_{1:1}+\epsilon}\cdot m} \leq \exp(-2m\epsilon^2). 
\end{align*}
Let $\event$ be the event that 
\begin{align*}
\sum\nolimits_{j\in[m]}\val_{ij} \geq \rbr{\dist_{1:1}+\epsilon}\cdot m
\end{align*}
for all bidders $i\in[n]$. 
By the union bound, with a total of $N$ bidders, we have $\prob{\event}\geq 1-N\cdot \exp(-2m\epsilon^2).$
Therefore, the expected revenue from pure bundling is at most the welfare from pure bundling, where the latter is upper bounded by
\begin{align*}
(1-N\cdot \exp(-2m\epsilon^2)) \cdot \rbr{\dist_{1:1}+\epsilon}\cdot m
+ mN\cdot \exp(-2m\epsilon^2)
\end{align*}
If $N=o(\exp(m))$, this upper bound converges to $\rbr{\dist_{1:1}+\epsilon}\cdot m$, 
which is strictly less than the optimal welfare for sufficiently small $\epsilon>0$ when $n\geq 2$.
\end{proof}

% \thmCompComplexityMHR*
\begin{theorem}[Competition complexity of {\bspa} against {\wel} for MHR items and a single bidder]
\label{thm:bspa mhr}
For $n=1$ bidder and any $m\geq 1$ asymmetric MHR items, the competition complexity of the bundle-based second-price auction against the welfare benchmark is at most $3$, i.e., $\bspa_{4} \geq \wel_1$.
\end{theorem}
\begin{proof}
By Theorem 3.2 of \citet{BMP-63}, the summation of values drawn from independent but possibly asymmetric MHR distributions is MHR. 
Therefore, by applying the argument in the proof of \cref{thm:welfare}, it is sufficient to consider the competition complexity of a single item with exponential distribution. 
Note that for exponential distributions, the revenue from $\bspa$ with $N$ bidders is $H_{N}-1$ where $H_N$ is the harmonic number, 
and the welfare is $1$. Therefore, when $N=4$, $H_N-1> 1$ and hence the competition complexity is $3$.
\end{proof}

\section{Single Bidder Setting: Competition Complexity of Bundling} 
\label{sec:cc_bundle}

In this section we explore the competition complexity of the bundle-based second-price auction ({\bspa}) in the single bidder setting, particularly for a small number of items (i.e., $m = 2, 3$). 

\begin{theorem}[Competition complexity of {\bspa} against {\cdw} for regular items and a single bidder]
\label{thm:bspa:3 items}
For $n = 1$ bidder and $m\in\{2, 3\}$ asymmetric regular items, the competition complexity of the bundle-based second-price auction against the duality benchmark is at most $m$, i.e., $\bspa_{1 + m} \geq \cdw_1$. 
\end{theorem}

We remark that while \Cref{thm:bspa:3 items} considers selling $m = 2$ or $m = 3$ items, the result of \citet{BK-96} (see \Cref{thm:BK}) can be viewed as our theorem with $m = 1$ regular item in the special case of $n = 1$ bidder. \Cref{bspa_multiple_n} in the previous section demonstrates why the single bidder setting ($n = 1$) is important: for $n \geq 2$ bidders, the competition complexity of {\bspa} is $\Omega(\exp(m))$, even for uniform distributions (which are MHR and thus regular).

In the remainder of this section, we formally prove \Cref{thm:bspa:3 items}. Our analysis follows a \emph{three-step approach}, with the key concept being an $m \times (m + 1)$ randomized quantile matrix $\quantMatrix$, where each entry is an i.i.d.\ draw from the uniform distribution $U[0, 1]$. \Cref{thm:bspa:3 items} is an immediate consequence of \Cref{coro:cdwapprox,prop:bspa_min_max,prop:cc_bspa_m_3} which we develop sequentially.

\xhdr{Step 1 (Upper bound on the duality benchmark):}  
In \Cref{subsec:CDW upper bound} (\Cref{prop:cdwapprox,coro:cdwapprox}), we upper bound the duality benchmark $\cdw_1(\dist)$ by the expectation of a function $\cdw_1(\quantMatrix)$ (formally defined in the corollary statement) over the randomized quantile matrix $\quantMatrix$: 
\begin{align*}
    \cdw_1(\dist) \leq \expect[\quantMatrix]{\cdw_1(\quantMatrix)}
\end{align*}

\xhdr{Step 2 (Revenue characterization of BSPA):}  
In \Cref{subsec:BSPA upper bound} (\Cref{prop:bspa_min_max}), we express the expected revenue of {\bspa} as the expectation of the expectation of another function $\bspa_{1+m}(\quantMatrix)$ over the randomized quantile matrix $\quantMatrix$:
\begin{align*}
     \bspa_{1 + m}(\dist) = \expect[\quantMatrix]{\bspa_{1 + m}(\quantMatrix)}
\end{align*}
where function $\bspa_{1+m}(\quantMatrix)$ is defined as the value of a two-player zero-sum game parameterized by the quantile matrix $\quantMatrix$.

\xhdr{Step 3 (Competition complexity analysis per quantile matrix):}  
Intuitively speaking, these first two steps establish a coupling between the duality benchmark and {\bspa} via the quantile matrix $\quantMatrix$. Consequently, in \Cref{subsec:competition complexity via quantile matrix} (\Cref{prop:cc_bspa_m_3}), we argue that when $m\in\{2, 3\}$, for every realized quantile matrix $\quantMatrix$,  
\begin{align*}
    \bspa_{1+m}(\quantMatrix) \geq \cdw_1(\quantMatrix)
\end{align*} 
Notably, the first two steps hold for any number of items ($m \geq 1$) and may be of independent interest for future research. 

\subsection{Upper Bounding the Duality Benchmark by Randomized Quantile Matrix}
\label{subsec:CDW upper bound}
In this section, we construct a function $\cdw_1(\quantMatrix)$ such that $\cdw_1(\dist) \leq \expect[\quantMatrix]{\cdw_1(\quantMatrix)}$ in \Cref{prop:cdwapprox} and \Cref{coro:cdwapprox}. 
\begin{lemma}[Upper bound of {\cdw} by randomized quantile profile] \label{prop:cdwapprox}
    For any $m\in \naturals$, let $\dist = \Pi_{j\in[m]} \dist_j$ be an arbitrary value distribution that is independent over $m$ regular items, where $\dist_j$ is the value distribution for each item $j\in[m]$. Consider the following process:
    \begin{itemize}
        \item 
         Draw $m+1$ quantiles $\quant'_1, \cdots, \quant'_{m+1}$ independently from the uniform distribution $U[0,1]$, and let $\quant'^*_1 \geq \quant'^*_2 \geq \cdots \geq \quant'^*_{m+1}$ be $\quant'_1, \cdots, \quant'_{m+1}$ sorted in decreasing order. 
        \item 
        Return $\frac{m}{m+1} \cdot (2 g_{\dist}(\quant'^*_2) + g_{\dist}(\quant'^*_3) + \cdots + g_{\dist}(\quant'^*_{m+1}))$, where the auxiliary function $g_{\dist}(\quant) \triangleq \frac{1}{m} \sum\nolimits_{j\in[m]} \dist_j^{-1}(\quant)$.
    \end{itemize}
     Then, the duality benchmark $\cdw_1(\dist)$ is at most the expected value of the return of the process:
     \begin{align*}
         \cdw_1(\dist) \leq \expect{\frac{m}{m+1} \cdot (2 g_{\dist}(\quant'^*_2) + g_{\dist}(\quant'^*_3) + \cdots + g_{\dist}(\quant'^*_{m+1}))}
     \end{align*}
\end{lemma}
\begin{proof}
    By \Cref{cor:duality_one_buyer}, we know that when the items are regular, $\cdw_1(\dist)$ is equal to expected value of the following process, which we call $\QUANT$:
    \begin{itemize} 
        \item 
         Draw $m$ quantiles $\quant_1, \cdots, \quant_{m}$ independently from the uniform distribution $U[0, 1]$ for the $m$ items. Namely, the value of the buyer for each item $j\in[m]$ is $\val_j = \dist_j^{-1}(\quant_j)$. Next, let $\tau$ be the sorting permutation of the values such that $\quant_{\tau(1)} \geq \quant_{\tau(2)} \geq \cdots \geq \quant_{\tau(m)}$.  
        \item 
        Return 
        $\plus{\virtual_{\tau(1)}(\dist_{\tau(1)}^{-1}(\quant_{\tau(1)}))} + \sum\nolimits_{j\in[2:m]} \dist_{\tau(j)}^{-1}(\quant_{\tau(j)})$.
    \end{itemize}
    Since the uniform distribution $U[0,1]$ is atomless, the probability that two quantiles drawn from $U[0,1]$ are identical is zero. Hence the above process of generating quantiles $\quant_1, \cdots, \quant_m$ is equivalent to first generating the first, second, $\dots$, $m$-th order statistics $\quant_1^* \geq \quant_2^* \geq \cdots \geq \quant_m^*$ from the uniform distribution $U[0,1]$, and then randomly assigning item indices to the order statistics. Hence, the expected value of $\QUANT$ is equivalent to the following process $\RANDQ$: 
    \begin{itemize}
        \item 
        Draw $m$ quantiles $\quant_1, \cdots, \quant_{m}$ independently from the uniform distribution $U[0, 1]$ and let $\quant^*_1 \geq \quant^*_2 \geq \cdots \geq \quant^*_m$ be $\quant_1, \cdots \quant_m$ sorted in descending order. 
        \item Draw permutation $\tau$ uniformly at random from the symmetry group $S_m$. For each index~$k\in[m]$, let item $\tau(k)$ have value $\val_{\tau(k)} = \dist_{\tau(k)}^{-1}(\quant_k^*)$.
        \item 
        Return $\plus{\virtual_{\tau(1)}(\dist_{\tau(1)}^{-1}(\quant^*_1))} + \sum\nolimits_{k\in[2:m]} \dist_{\tau(k)}^{-1}(\quant^*_{k})$.
    \end{itemize}
    By first taking the expectation over $\quant^*_1 \geq \quant^*_2 \geq \cdots \geq \quant^*_m$, and then taking expectation over the permutation $\tau$ (that takes a uniformly random value from the symmetry group $S_m$), process $\RANDQ$ returns expected value: 
    \begin{align*}
        &\expect[{\vec{\quant^*}},\tau] {
        \plus{\virtual_{\tau(1)}(\dist_{\tau(1)}^{-1}(\quant^*_1))} + \sum\nolimits_{k\in[2:m]} \dist_{\tau(k)}^{-1}(\quant^*_{k})
        }
        \\
         = {} &
        \expect[\vec{\quant^*}]{
        \expect[\tau]{
            \plus{\virtual_{\tau(1)}(\dist_{\tau(1)}^{-1}(\quant^*_1))}
        } + \sum\nolimits_{k\in[2:m]} \expect[\tau]{
            \dist_{\tau(k)}^{-1}(\quant^*_{k})
        }} \\
         = {} &\expect[\vec{\quant^*}] {\frac{1}{m} \sum\nolimits_{j\in[m]} \plus{\virtual_j(\dist_{j}^{-1}(\quant^*_1))} + \sum\nolimits_{k\in[2:m]} \frac{1}{m}  \sum\nolimits_{j\in[m]} \dist_{j}^{-1}(\quant^*_k)}
        \\
     = {} &\frac{1}{m} \sum\nolimits_{j\in[m]} \expect[\vec{\quant^*}] {\plus{\virtual_j(\dist_{j}^{-1}(\quant^*_1))}} + \sum\nolimits_{k\in[2:m]} \expect[\vec{\quant^*}] {g_{\dist}(\quant^*_k)}. 
    \end{align*}
    where auxiliary function $g_{\dist}(\cdot)$ is defined in the lemma statement.
     
    We now claim that the expected value from process $\RANDQ$ is at most the expected value from the following process $\RANDQPlus$: 
    \begin{itemize} 
        \item 
         Draw $m+1$ quantiles $\quant'_1, \cdots, \quant'_{m+1}$ independently from the uniform distribution $U[0,1]$, and let $\quant'^*_1 \geq \quant'^*_2 \geq \cdots \geq \quant'^*_{m+1}$ be $\quant'_1, \cdots, \quant'_{m+1}$ sorted in decreasing order. 
        \item 
        Uniformly at random select an index $k\primed \in [m+1]$, let $\quant^*_1 \geq \cdots \geq \quant^*_m$ be the resulting sequence after deleting $\quant'^*_{k\primed}$ from $\quant'^*_1 \cdots \quant'^*_{m+1}$. 
        \item 
        Return $g_{\dist}(\quant'^*_2)+ \sum\nolimits_{k\in[2:m]} g_{\dist}(\quant^*_k)$. 
    \end{itemize}
    Notice that since we uniformly drop an index from $\quant'^*_1, \cdots \quant'^*_{m+1}$, random quantiles $\quant^*_1, \cdots, \quant^*_m$ in process $\RANDQ$ and $\RANDQPlus$ has the exact same distribution,  and hence $g_{\dist}(\quant^*_2), \cdots g_{\dist}(\quant^*_{m})$ in both processes have the same expected value. It remains to compare the expected value of 
    \begin{align*}
        \frac{1}{m} \sum\nolimits_{j\in[m]} \plus{\virtual_j(\dist_{j}^{-1}(\quant^*_1))}
        % \text{ in } \RANDQ
        \;\;
        \mbox{and}
        \;\;
        g_{\dist}(\quant'^*_2)
    \end{align*}
    in process $\RANDQ$ and $\RANDQPlus$, respectively.
    Firstly, notice that  $\expect{\plus{\virtual_j(\dist_j^{-1}(\quant^*_1))}}$ is also equal to the expected optimal revenue when selling item $j$ to $m$ i.i.d.\ bidders, i.e., $\rev_{m}(\dist_j)$. Since item $j$ is regular, by \Cref{thm:BK}, $\rev_{m}(\dist_j)$ is at most 
    the revenue from selling item $j$ to $m + 1$ i.i.d.\ bidders with the second-price auction, i.e.,  
    $\vcg_{m+1}(\dist_j) = \expect{\dist_j^{-1}(\quant'^*_2)}$.  
    Summing over all items $j \in [m]$, we have 
    \begin{align*}
        \frac{1}{m} \sum\nolimits_{j\in[m]} \expect{\plus{
        \virtual_j(\dist_{j}^{-1}(\quant^*_1))}}
        \leq \frac{1}{m} \sum\nolimits_{j\in[m]}\expect {\dist_j^{-1}(\quant'^*_2)} = 
        \expect{g_{\dist}(\quant'^*_2)}
    \end{align*}
    Finally, we will prove that the expected value from $\RANDQPlus$ is equal to the expected value from the process stated in the lemma statement. Since $\quant^*_1, \cdots, \quant^*_m$ is a result from deleting one of $\quant'^*_1, \cdots, \quant'^*_{m+1}$ uniformly at random, by a simple counting argument, in process $\RANDQPlus$, the expected value of 
    \begin{align*}
        g_{\dist}(\quant^*_2) + \cdots + g_{\dist}(\quant^*_{m})
    \end{align*}
    is equal to the expected value of 
    \begin{align*}
        \frac{m-1}{m+1} \cdot g_{\dist}(\quant'^*_2) + \frac{m}{m+1} \cdot \left(g_{\dist}(\quant'^*_3) + \cdots + g_{\dist}(\quant'^*_{m+1})\right)
    \end{align*}
    Specifically, for each $k = 2, \cdots, m$, random quantiles $\quant^*_k = \quant'^*_k$ when the index $k\primed$ deleted is at least $k+1$. This occurs with probability $\frac{m+1-k}{m+1}$. Otherwise, $\quant^*_k = \quant'^*_{k+1}$. Thus 
    \begin{align*}
        \expect{
            \sum\nolimits_{k\in[2:m]} g_{\dist}(\quant^*_k)
        } 
        &= \expect{
            \sum\nolimits_{k\in[2:m]} \left(\frac{m+1-k}{m+1} \cdot g_{\dist}(\quant'^*_k) + \frac{k}{m+1} \cdot  g_{\dist}(\quant'^*_{k+1})\right)}
         \\
        &= \expect{
            \frac{m-1}{m+1} \cdot g_{\dist}(\quant'^*_2) + 
            \sum\nolimits_{k\in[3:m + 1]}\frac{m}{m+1} \cdot g_{\dist}(\quant'^*_k)}
    \end{align*}
    Hence, we conclude that 
    \begin{align*}
        \expect{g_{\dist}(\quant'^*_2) + \sum\nolimits_{k\in[2:m]}g_{\dist}(\quant^*_k)}
        =
        \expect{\frac{m}{m + 1}\cdot 
        \left(
        2g_{\dist}(\quant'^*_2)
        +
        \sum\nolimits_{k\in[3:m+1]}
        g_{\dist}(\quant'^*_k)
        \right)}
    \end{align*}
    which concludes the proof of the lemma statement as desired.
\end{proof}

\begin{lemma}[Upper bound of {\cdw} by randomized quantile matrix]
\label{coro:cdwapprox}
    For any $m\in \naturals$, let $\dist = \Pi_{j\in[m]} \dist_j$ be an arbitrary value distribution that is independent over $m$ regular items, where $\dist_j$ is the value distribution for each item $j\in[m]$. Consider the following process:
    \begin{itemize}
        \item 
        Generate the following $m\times (m + 1)$ randomized quantile matrix $\quantMatrix$:
        for each row $i\in[m]$, draw $m+1$ quantiles $\quant_{i, 1}, \cdots, \quant_{i, m+1}$ independently from the uniform distribution $U[0,1]$ and set matrix entry $\quantMatrix[i, k] = \quant_{i, k}$ for each $k\in[m + 1]$.
        \item Sort $m + 1$ quantiles in each row $i\in[n]$ of $\quantMatrix$ so that $\quantMatrix[i, 1] \geq \quantMatrix[i, 2] \geq \cdots \geq \quantMatrix[i,m+1]$. Denote the row sorted version of $\quantMatrix$ as $\quantMatrix^*$. 
        \item 
        Define weight $x_k = \frac{2}{m+1}$ for $k = 2$ and $x_k = \frac{1}{m+1}$ for $k = 3, \cdots, m+1$. 
        \item 
        Return $\cdw_1(\quantMatrix)\triangleq  \frac{1}{m+1}\cdot \sum\nolimits_{i\in[m]}
        \left(
        2g_{\dist}(\quantMatrix^*[i, k])
        +
        \sum\nolimits_{k\in[3:m +1]} g_{\dist}(\quantMatrix^*[i, k])\right)$, where the auxiliary function $g_{\dist}(\quant) \triangleq \frac{1}{m} \sum\nolimits_{j\in[m]} \dist_j^{-1}(\quant)$.
    \end{itemize}
    Then, the duality benchmark $\cdw_1(\dist)$ is at most the expected value of the return of the process:
     \begin{align*}
         \cdw_1(\dist) \leq \expect[\quantMatrix]{\cdw_1(\quantMatrix)}
     \end{align*}
\end{lemma}
\begin{proof}
    The process described in the lemma statement is equivalent to averaging over $m$ independent processes as described in \Cref{prop:cdwapprox}, which preserves the expected value. 
\end{proof}

\subsection{Revenue Characterization of the BSPA by Randomized Quantile Matrix}
\label{subsec:BSPA upper bound}
In this section, we construct function $\bspa_{1+m}(\quantMatrix)$ such that $\bspa_{1+m}(\dist) = \expect[\quantMatrix]{\bspa_{1+m}(\quantMatrix)}$ in \Cref{prop:bspa_min_max}. The construction of function $\bspa_{1+m}(\quantMatrix)$ relies on the two-player zero-sum game (\Cref{def:two player zero sum game}).

Before presenting the formal definition, we first provide the intuition behind the construction of the two-player game. Given any $m \times (m + 1)$ quantile matrix $\quantMatrix$, it can be interpreted as representing a setting where $m$ items are sold to $m + 1$ bidders. In this interpretation, the entry $\quantMatrix[i, j]$ denotes the value quantile of bidder $j$ for item $i$. That is, each row corresponds to an item, and each column corresponds to a bidder.\footnote{In contrast, in the upper bound analysis of the duality benchmark (\Cref{prop:cdwapprox,coro:cdwapprox}), we interpret the $m$ rows of the quantile matrix as $m$ bidders and the $m + 1$ columns as $m$ items, after randomly dropping one column.}
Since the revenue in the {\bspa} corresponds to the second-highest total value for the grand bundle, we model this using a two-player zero-sum game. In this game, an adversary (the min-player) first removes one column (i.e., one bidder) from the $m + 1$ columns. Then, an optimizer (the max-player) selects one of the remaining columns (i.e., one of the remaining bidders) and collects that bidder’s total value as the payment. Clearly, the optimal strategy of the adversary is to remove the column corresponding to the bidder with the highest total value, while the best response of the optimizer is to select the column corresponding to the second-highest total value. Note that we can convexify the action spaces of both players without loss of generality.

While this game construction provides a valid characterization of the revenue in {\bspa}, it is insufficient for establishing the desired comparison of $\bspa_{1 + m}(\quantMatrix) \geq \cdw_1(\quantMatrix)$ in \Cref{subsec:competition complexity via quantile matrix}. The issue is that, under this construction, there exist quantile matrices $\quantMatrix$ for which the revenue is too low.\footnote{Loosely speaking, this corresponds to value profiles where, for example, the second-highest bidder for the grand bundle has zero value, resulting in significantly lower ex post revenue in {\bspa} compared to the benchmark.}
To address this issue, we introduce additional randomness into the construction of $\bspa_{1 + m}(\quantMatrix)$. Similar to the approach in \Cref{prop:cdwapprox}, and leveraging the fact that our analysis is conducted in the quantile space, we apply random permutations to the entries of the quantile matrix $\quantMatrix$. However, these permutations must be constructed with care. The formal definition of our two-player zero-sum game follows below.

\begin{definition}[Zero-sum game formulation]
\label{def:two player zero sum game}
    Given any $m \in\naturals$ and $m\times (m + 1)$ quantile matrix $\quantMatrix\in[0, 1]^{m\times (m + 1)}$, define the following two-player zero-sum game:
    \begin{enumerate}
        \item The nature first applies a uniform random permutation to each row $i \in [m]$ of the quantile matrix $\quantMatrix$. Denote this row-wise permutation as $\sigma_i: [m+1] \rightarrow [m+1]$. Then, nature applies a uniform random permutation over the rows, denoted by $\tau: [m] \rightarrow [m]$. Let $\sigma = \prod_{i\in[m]} \sigma_i$ represent the combined effect of all row-wise permutations. (An illustration of the permutations $\sigma_i$ and the composite permutation $\sigma$ is provided in \Cref{fig:sigma}.)
        \item Given the realized permutations $\sigma, \tau$, the min-player (aka., adversary) picks a distribution over columns, denoted by $A \in \Delta([m+1])$.
        \item Given the realized permutations $\sigma, \tau$ and min-player's action $A \in \Delta([m+1])$, the max-player (aka., optimizer) picks a distribution over columns $B \in \Delta([m+1])$. 
    \end{enumerate}
    Given the realized permutations $\sigma, \tau$, and two players' actions $A,B$, the payoff is 
    \begin{align*}
        \Payoff{A, B, \sigma, \tau}
        \triangleq 
        \sum\nolimits_{j\in[m + 1]}
        \min\{B(j), 1 - A(j)\} \cdot \sum\nolimits_{i\in[m]} \dist_{i}^{-1}\left(\quantMatrix[\tau(i),\sigma_{\tau(i)}(j)]\right)
    \end{align*}
    Define function $\bspa_{1+m}(\quantMatrix)$ as the value of this two-player zero-sum game:
    \begin{align*}
        &\bspa_{1+m}(\quantMatrix) \triangleq 
        \expect[\sigma,\tau]{\max_{B\in\Delta([m+1])}
        \min_{A\in\Delta([m+1])}\Payoff{A, B, \sigma,\tau}}
    \end{align*}
\end{definition}

\begin{figure}
% \noindent
\begin{minipage}{0.48\textwidth}   
    \centering
\begin{tikzpicture}
    \def\m{4}  % Number of rows
    \def\n{4} % Number of columns
    \def\permrow{3} % Row to be permuted (1-based index)
    % \def\permrownum{i}
    % Define the matrix using the matrix of nodes approach
    \matrix[matrix of nodes,
        nodes={draw, minimum width=1.5cm, minimum height=1cm, anchor=center},
        column sep=-\pgflinewidth, row sep=-\pgflinewidth] (Q) 
    {
        $q_{1,1}$  & $q_{1,2}$ & $\cdots$  & $q_{1,m+1}$  \\
        $ $   & $\vdots$    &   $ $        & $\vdots$    \\
        $q_{i,1}$  & $q_{i,2}$  & $\cdots$  & $q_{i,m+1}$  \\
        $ $   & $\vdots$    &   $ $        & $\vdots$    \\
        $q_{m,1}$  & $q_{m,2}$  & $\cdots$  & $q_{m,m+1}$  \\
    };~~~

    % Draw permutation arrows within the selected row
    \foreach \j in {1,...,\n} {
        \pgfmathparse{int(mod(\j+1,\n)+1)} % Simple 
        \let\nextj\pgfmathresult
        \ifnum\nextj>1
            \draw[->, orange!80, thick] ([yshift=0.2cm]Q-\permrow-\nextj.south) -- ([yshift=0.2cm]Q-\permrow-\j.south);
        \fi
    }

    % Label the matrix
    \node[below] at ([xshift=-0.5cm]Q-5-3.south) {$Q$};
    
    % Label the permutation (less saturated color)
    \node[below, orange!90, font=\large] at ([yshift=0.05cm, xshift=-0.3cm]Q-\permrow-2.south east) {$\sigma_{i}$};

\end{tikzpicture}
% \captionof{figure}{Illustration of $\sigma_{i}$ (permutes $i^{th}$ row)} 
% \label{fig:sigmai}
\end{minipage}
\hfill
\begin{minipage}{0.48\textwidth} 
    \centering
\begin{tikzpicture}
    \def\m{5}  % Number of rows
    \def\n{4} % Number of columns
    
    % Define the matrix using the matrix of nodes approach
    \matrix[matrix of nodes,
        nodes={draw, minimum width=1.5cm, minimum height=1cm, anchor=center},
        column sep=-\pgflinewidth, row sep=-\pgflinewidth] (Q) 
    {
        $q_{1,1}$  & $q_{1,2}$  & $\cdots$  & $q_{1,m+1}$  \\
        $\vdots$   & $\vdots$   &           & $\vdots$     \\
        $q_{i,1}$  & $q_{i,2}$  & $\cdots$  & $q_{i,m+1}$  \\
        $\vdots$   & $\vdots$   &           & $\vdots$     \\
        $q_{m,1}$  & $q_{m,2}$  & $\cdots$  & $q_{m,m+1}$  \\
    };

    % Draw permutation arrows within every row
    \foreach \i in {1,...,\m} {
        \foreach \j in {1,...,\n} {
            \pgfmathparse{int(mod(\j+1,\n)+1)} % Simple cyclic shift
            \let\nextj\pgfmathresult
            \ifnum\nextj>1
                \draw[->, orange!80, thick] ([yshift=0.2cm]Q-\i-\nextj.south) -- ([yshift=0.2cm]Q-\i-\j.south);
            \fi
        }
    }
    
    % Label the matrix
    \node[below] at ([xshift=-0.5cm]Q-5-3.south) {$Q$};
    
    % Label the permutation (less saturated color)
    \node[below, orange!90, font=\large] at ([yshift=0.05cm, xshift=-0.3cm]Q-3-2.south east) {$\sigma$};

\end{tikzpicture}
% \captionof{figure}{Illustration of $\sigma$ (permutes all rows)}
% \label{fig:sigma}
\end{minipage}
\caption{Illustration of permutations $\sigma_{i}$ (permutes the $i$-th row) and $\sigma$ (permutes all rows) in \Cref{def:two player zero sum game}.} \label{fig:sigma}
\end{figure}
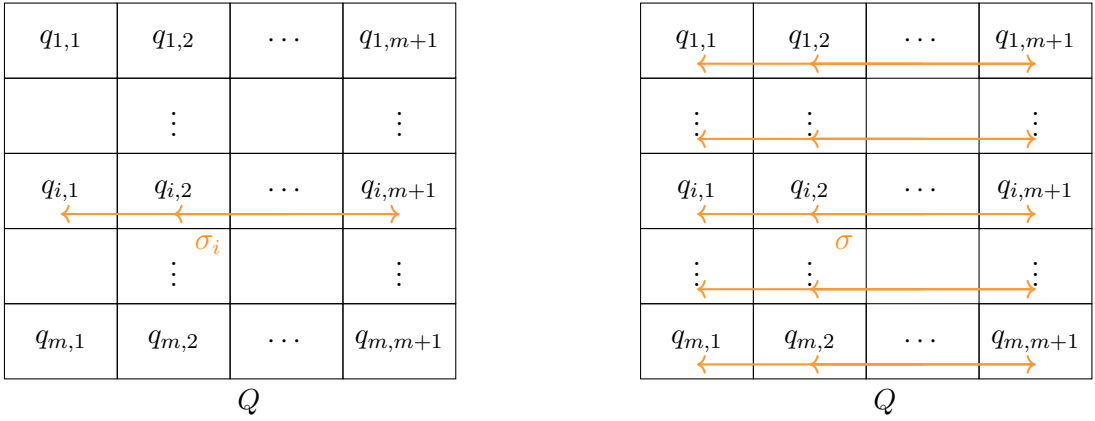
 
\begin{restatable}[Revenue characterization of {\bspa} by randomized quantile matrix]{lemma}{propBspaMinMax} \label{prop:bspa_min_max}
For any $m\in \naturals$, let $\dist = \Pi_{j\in[m]} \dist_j$ be an arbitrary value distribution that is independent over $m$ regular items, where $\dist_j$ is the value distribution for each item $j\in[m]$. 
The expected revenue of the {\bspa} with $m + 1$ bidders equals the expectation of $\bspa_{1 + m}(\quantMatrix)$ over $m\times (m + 1)$ randomized quantile matrix $\quantMatrix$, where each entry $\quantMatrix$ is drawn independently from the uniform distribution $U[0, 1]$:
\begin{align*}
    \bspa_{1 + m}(\dist) = \expect[\quantMatrix]{\bspa_{1 + m}(\quantMatrix)}
\end{align*}
\end{restatable}
\begin{proof}
Fix any realized permutations $\sigma,\tau$. Define auxiliary variable $s_j$ for each column $j\in[m + 1]$:
\begin{align*}   
    s_j \triangleq  \sum\nolimits_{i\in[m]} F_{i}^{-1}\left(\quantMatrix[\tau(i),\sigma_{\tau(i)}(j)]\right)
\end{align*}
% Since each entry in $\quantMatrix$ is drawn i.i.d.\ from the uniform distribution $U[0, 1]$, for any fixed $\sigma$ and $\tau$, $\quantMatrix[\tau(i),\sigma_{\tau(i)}(j)]$ is still drawn i.i.d from $U[0,1]$. 
Without loss of generality, relabel the columns so that $s_1 \;\ge\; s_2 \;\ge\; \cdots \;\ge\; s_{m+1}$. Note that $s_j$ could be viewed as the $j$-th highest bundle value among $m+1$ bidders, given that each bidder $j$'s value for item $i$ is $F_{i}^{-1}\left(\quantMatrix[\tau(i),\sigma_{\tau(i)}(j)]\right)$. 
By definition of the {\bspa}, its expected revenue can be expressed as 
\begin{align*}
    \bspa_{1+m}(\dist) = \expect[\sigma,\tau]{\expect[\quantMatrix]{s_2}}
    =
    \expect[\quantMatrix]{\expect[\sigma,\tau]{s_2}}
\end{align*}
where we view $s_2$ as a random variable correlated with $\sigma,\tau$ and $\quantMatrix$.
% Here we use the fact that each entry in $\quantMatrix$ is drawn i.i.d.\ from the uniform distribution $U[0, 1]$, and thus for any fixed $\sigma$ and $\tau$, $\quantMatrix[\tau(i),\sigma_{\tau(i)}(j)]$ is still drawn i.i.d.\ from $U[0,1]$, which further ensures that relabeling the columns does not change.

Next we claim that for any realized permutations $\sigma,\tau$, the value of the two-player zero-sum game is $s_2$, i.e., 
\begin{align*}
    s_2 \equiv \max_{B\in\Delta([m+1])}
        \min_{A\in\Delta([m+1])}\Payoff{A, B, \sigma,\tau}
\end{align*}
where the payoff in the right-hand side is defined (in \Cref{def:two player zero sum game}) as 
\begin{align*}
    \Payoff{A, B, \sigma, \tau}
        \triangleq 
        \sum\nolimits_{j\in[m + 1]}
        \min\{B(j), 1 - A(j)\} \cdot s_j
\end{align*}
Our argument contains two parts. First, we argue that $\max_B\min_A \Payoff{A, B, \sigma,\tau}$ is at most $s_2$. To see this, consider the strategy $A\primed$ of the adversary (min-player) which puts its entire probability mass on the column with the largest sum, i.e.,
\begin{align*}
  A\primed(1) = 1
  \;\;
  \mbox{and}
  \;\;
  A\primed(j) = 0 \text{ for } j\in[2:m+1],
\end{align*}
Consequently, column~1 is effectively ``removed'' from the optimizer (max-player)’s choices.  As a result, the maximum that the optimizer can get is achieved by placing all probability on the best \emph{remaining} column, whose sum is $s_2$, i.e., 
\begin{align*}  
  s_2 \equiv 
  \max_{B\in\Delta([m+1])}
        \sum\nolimits_{j\in[m + 1]}
        \min\{B(j), 1 - A\primed(j)\} \cdot s_j
\end{align*}
In summary, under the adversary’s ``pure'' strategy $A\primed$ that deletes the highest-sum column, the payoff is at most $s_2$.

Conversely, we argue that $\max_B\min_A \Payoff{A, B, \sigma,\tau}$ is at least $s_2$. To see this, note that the optimizer has a simple ``pure'' strategy to ensure a payoff of at least $s_2$.  If the adversary places all probability on column~1 (the largest sum), then column~1 is removed, so the optimizer picks column~2 and obtains $s_2$.  If the adversary places positive probability on any other column (with sum at most $s_2$), then the adversary's weight $A_1$ on column~1 is strict less than 1, i.e., $A(1) < 1$. In this case, the optimizer can select $B(1) = 1 - A(1)$ and $B(2) = A(1)$ weight on $s_2$ to achieve a payoff at least $s_2$. This proves our claim. 

Putting the two pieces together, we obtain 
\begin{align*}
     \bspa_{1+m}(\dist) 
    =
    \expect[\quantMatrix]{\expect[\sigma,\tau]{s_2}}
    =
    \expect[\quantMatrix]{\bspa_{1+m}(\quantMatrix)}
\end{align*}
which finishes the proof of \Cref{prop:bspa_min_max} as desired.
\end{proof}

\subsection{Competition Complexity Analysis per Quantile Matrix}
\label{subsec:competition complexity via quantile matrix}
In this section, we establish the ``competition complexity'' between functions $\bspa_{1 + m}(\quantMatrix)$ and $\cdw_1(\quantMatrix)$ developed in the previous two subsections for each realized quantile matrix $\quantMatrix$ separately.

\begin{restatable}[Competition complexity of {\bspa} against {\cdw} per quantile matrix]{lemma}{propBspaThree}\label{prop:cc_bspa_m_3}
For any $m\in \{2, 3\}$, let $\dist = \Pi_{j\in[m]} \dist_j$ be an arbitrary value distribution that is independent over $m$ regular items, where $\dist_j$ is the value distribution for each item $j\in[m]$. Fixing any $m \times (m+1)$ quantile matrix $\quantMatrix\in[0,1]^{m\times(m+1)}$, it satisfies that
\begin{align*}
    \bspa_{1 + m}(\quantMatrix) \geq \cdw_1(\quantMatrix)
\end{align*}
\end{restatable}

Recall that the function $\bspa_{1 + m}(\quantMatrix)$ is defined via a two-player zero-sum game (\Cref{def:two player zero sum game}) between an adversary (min-player) and an optimizer (max-player). We explicitly construct the optimizer's strategies for different cases. This construction yields a payoff for the optimizer that serves as a valid lower bound on $\bspa_{1 + m}(\quantMatrix)$. We then compare this lower bound with the benchmark $\cdw_1(\quantMatrix)$ to establish the desired result to prove \Cref{prop:cc_bspa_m_3}.

\begin{proof}[Proof of \Cref{prop:cc_bspa_m_3} with $m = 2$ bidders.]
    We defer a similar but more involved analysis for $m = 3$ in \Cref{app:sec:cc_bundle}.
    Suppose there are $m = 2$ bidder, quantile matrix $\quantMatrix$ has a size $2 \times 3$. Let $\quantMatrix^*$ be an associated matrix where $\quantMatrix^*[i, j]$ is the $j$-th highest order statistic among $\quantMatrix[i, 1], \cdots, \quantMatrix[i, m+1]$.

    Consider the two-player zero-sum game defined in \Cref{def:two player zero sum game}.
    For any row permutation $\sigma$, there are two possibilities: 1) $\quantMatrix^*[1, 3]$ and $\quantMatrix^*[2, 3]$ are in the same column, and 2) $\quantMatrix^*[1, 3]$ and $\quantMatrix^*[2, 3]$ are in different columns. Without loss of generality, we assume the adversary strategy $A$ is always an extreme point of $\Delta([m+1])$, since there always exists an optimal adversary strategy of this form (as witnessed in the proof of \Cref{prop:bspa_min_max}). We next compute the payoff that the optimizer gets from a specific strategy, which lower bounds $\bspa_{3}(\quantMatrix)$.
    
    In the first case, no matter which column the adversary selects, the optimizer will always be able to select a column $j$ where for all item $i \in [2]$, $\quantMatrix[\tau(i),\sigma_{\tau(i)}(j)] \geq \quantMatrix^*[\tau(i), 2]$. In the second case, no matter which column the adversary selects, the optimizer can selected each remaining column with $1/2$ probability. This means that for any item $i \in [2]$, the optimizer select a column $j$ with $\quantMatrix[\tau(i),\sigma_{\tau(i)}(j)] \geq \quantMatrix^*[\tau(i), 2]$ with at least $1/2$ probability. The probability (over $\sigma$) that the first case happens is $1/3$, therefore for any item $i \in [2]$, the probability of the player selecting a column~$j$ where  $\quantMatrix[\tau(i),\sigma_{\tau(i)}(j)] \geq \quantMatrix^*[\tau(i), 2]$ is at least $1 \cdot 1/3 + 1/2 \cdot 2/3 = 2/3$. 
    Consequently, the optimizer's expected payoff from the above described strategy is at least 
    \begin{align*}
        &\expect[\tau]{\sum\nolimits_{i\in[2]} \frac{2}{3} \cdot \dist_{i}^{-1}(\quantMatrix^*[\tau(i), 2]) + \frac{1}{3} \cdot \dist_{i}^{-1}(\quantMatrix^*[\tau(i), 3])}
        \\
        = {} & \sum\nolimits_{i\in[2]} \frac{2}{3} \cdot \expect[\tau]{\dist_{i}^{-1}(\quantMatrix^*[\tau(i), 2])} + \frac{1}{3} \cdot \expect[\tau]{\dist_{i}^{-1}(\quantMatrix^*[\tau(i), 3])} \\
        = {} & \sum\nolimits_{i\in[2]} \frac{2}{3} \cdot g_{\dist}(\quantMatrix^*[i, 2]) + \frac{1}{3} \cdot g_{\dist}(\quantMatrix^*[i, 3])
        \\
        = {} & \cdw_1(\quantMatrix) 
    \end{align*}
    where the last equality holds due to the definition of function $\cdw_1(\quantMatrix)$ in \Cref{coro:cdwapprox}. This completes the proof of the lemma statement with $m = 2$ bidders as desired.
\end{proof}

\section{Single Bidder Setting: Constant Approximation of Bundling} \label{sec:approx_bundle}

In this section, we establish constant-factor revenue approximations for selling the grand bundle under the assumptions that items are either regular or MHR. In \Cref{subsec:brev approx}, we study the optimal pricing of the grand bundle ({\brev}). In \Cref{subsec:bspa approx}, we study the bundle-based second-price auction ({\bspa}). Unlike the previous section, which focused on settings with a small number of items, all results in this section hold for an arbitrary number of items, i.e., for any $m \in \naturals$.

\subsection{Constant Approximation by the Optimal Grand Bundle Pricing}  
\label{subsec:brev approx}

In this section, we study the optimal pricing of the grand bundle ({\brev}) for a single bidder. We begin by focusing on the setting with regular items. Our main result is as follows.

\begin{restatable}[Approximation of {\brev} against {\rev} for regular items and a single bidder]{theorem}{crEighteenBrevGegRev}
\label{cor:18brev_geq_rev}
For $n = 1$ bidder and $m \geq 1$ asymmetric regular items, the optimal expected revenue from selling all items as a grand bundle is an 18-approximation to the Bayesian optimal revenue, i.e., $18\cdot \brev_1 \geq \rev_1$. 
\end{restatable}

We first sketch the proof of \Cref{cor:18brev_geq_rev}.
Our analysis builds on a revenue approximation guarantee established by \citet{BILW-14} (\Cref{thm:six_approx}), which shows that combining two mechanisms--selling items separately ({\srev}) and selling all items as a grand bundle ({\brev})--yields a constant-factor approximation to the Bayesian optimal revenue. Importantly, this guarantee holds even when the value distributions are independent but irregular (i.e., not regular). Intuitively, {\srev} and {\brev} address different hard instances, and both are necessary to ensure a constant approximation in the presence of irregular distributions \citep{HN-12,BILW-14}.

In contrast, as shown in \Cref{cor:18brev_geq_rev}, when the value distributions are regular, selling items as a grand bundle ({\brev}) alone suffices to achieve a constant-factor approximation. Specifically, we prove in \Cref{pr:4brev_geq_srev} that under regular distributions, {\brev} is a 4-approximation to {\srev} via a simple analysis based on revenue curves. \Cref{cor:18brev_geq_rev} then follows directly from \Cref{thm:six_approx,pr:4brev_geq_srev}.

\begin{lemma}[\citealp{BILW-14}]\label{thm:six_approx}
For $n = 1$ bidder and $m \geq 1$ asymmetric items, 
the Bayesian optimal revenue is at most twice the optimal expected revenue from selling all items as a grand bundle plus four times the optimal expected revenue from selling all items separately, i.e., $2\cdot \brev_1 + 4\cdot \srev_1 \geq \rev_1$.
\end{lemma}

\begin{restatable}[Approximation of {\brev} against {\srev} for regular items and a single bidder]{lemma}{prFourBrevGeqSRev}
\label{pr:4brev_geq_srev}
For $n = 1$ bidder and $m \geq 1$ asymmetric regular items, the optimal expected revenue from selling all items as a grand bundle is a 4-approximation to the optimal expected revenue from selling all items separately, i.e., $4\cdot \brev_1 \geq \srev_1$.\footnote{In particular, this 4-approximation can be attained by setting the grand bundle price to half of the optimal expected revenue from selling the items separately, i.e., $\srev_1/2$.}
\end{restatable}

We emphasize that the regularity assumption in \Cref{pr:4brev_geq_srev} is essential. As noted above, a constant-factor approximation of {\brev} to {\srev} does not hold in general for irregular distributions \citep{HN-12,BILW-14}. The proof of this 4-approximation relies on the following technical lemma.

\begin{lemma}
\label{lem:brev approx worst dist}
    For any univariate distribution $\dist$, construct an auxiliary two-point distribution $\dist\primed$ with CDF $\cdf\primed(\val) = 0$ for $\val \in (-\infty, 0]$, $\cdf\primed(\val) = 1/2$ for $\val \in (0, \rev_1(\dist)]$, and $\cdf\primed(\val) = 1$ for $\val\in (\rev_1(\dist),\infty)$, where $\rev_1(\dist)$ is the Bayesian optimal revenue for selling a single item to a single bidder with value distribution $\dist$. If distribution $\dist$ is regular, it first-order stochastically dominates the constructed distribution $\dist\primed$, i.e., for every value $\val\in\reals$, CDFs $\cdf(\val) < \cdf\primed(\val)$.
\end{lemma}

We illustrate \Cref{lem:brev approx worst dist} in \Cref{fig:brev approx worst dist}. The proof is based on a straightforward geometric interpretation of the regularity condition, using the corresponding revenue curve defined as follows.

\begin{figure}
    \centering
    \begin{tikzpicture}[scale=1, transform shape]
\begin{axis}[
axis line style=gray,
axis lines=middle,
xtick={0, 0.21, 0.5, 0.75, 1},
ytick={0, 0.5, 1},
xticklabels={0, $\quant^*$, $\frac{1}{2}$, $\quant$,  1},
yticklabels={0, $\frac{1}{2}\cdot \rev_1(\dist)$, $\rev_1(\dist)$},
xmin=0,xmax=1.1,ymin=-0.05,ymax=1.2,
width=0.8\textwidth,
height=0.5\textwidth,
samples=500]

\addplot[dotted, gray, line width=0.3mm] (0, 1) -- (1, 1) -- (1, 0);
\addplot[dotted, gray, line width=0.3mm] (0.21, 0) -- (0.21, 1);

\addplot[dotted, gray, line width=0.3mm] (0.21, 0) -- (0.21, 1);

\addplot[dotted, gray, line width=0.3mm] (0, 0.5) -- (0.5, 0.5) -- (1, 1);
\addplot[dotted, gray, line width=0.3mm] (0.75, 0.75) -- (0.75, 0);

\addplot[domain=0:1, black!100!white, line width=0.5mm, smooth] coordinates {
    (0, 0)
    (0.1, 0.85)
    (0.2, 1)
    (0.75, 0.75)
    (1, 0)
};

\addplot[line width=0.5mm, dashed] (0, 0) -- (0.5, 0.5) -- (0.5, 0) -- (1, 0);
\end{axis}

\end{tikzpicture}
    \caption{Graphical illustration of \Cref{lem:brev approx worst dist} and its analysis. The black solid curve is the concave revenue curve $\revcurve$ induced from regular distribution $\dist$. The black dashed curve is the revenue curve~$\revcurve\primed$ induced from the constructed auxiliary two-point distribution $\dist\primed$. Due to the concavity of revenue curve $\revcurve$, probability $\quant \triangleq  \prob[\val\sim\dist]{\val \geq \rev_1(\dist)} \geq 1/2$.}
    \label{fig:brev approx worst dist}
\end{figure}

\begin{lemma}[\citealp{BR-89}]
    A univariate distribution $\dist$ is regular if and only if its induced revenue curve $\revcurve:[0,1]\rightarrow\reals_+$ defined as
    \begin{align*}
        \revcurve(\quant)\triangleq \quant \cdot \cdf^{-1}(1-\quant)
    \end{align*}
    is concave in $\quant\in[0, 1]$.
\end{lemma}
\begin{proof}[Proof of \Cref{lem:brev approx worst dist}]
    Due to the construction of auxiliary two-point distribution $\dist\primed$, it suffices to show that for regular distribution $\dist$, 
    \begin{align*}
        \prob[\val\sim\dist]{\val \geq \rev_1(\dist)} \geq \frac{1}{2}
    \end{align*}
    We denote the probability on the left-hand side by $\quant$.
    We prove $\quant \geq 1/2$ by analyzing the revenue curve $\revcurve$ induced from distribution $\dist$. Since distribution $\dist$ is regular, its revenue curve $\revcurve$ is concave, which further implies that
    \begin{align*}
        \revcurve(\quant) \geq \frac{1-\quant}{1- \quant^*}\revcurve(\quant^*) + \frac{\quant - \quant^*}{1 - \quant^*}\revcurve(1) 
        \;\;
        \Longrightarrow
        \;\;
        \quant \cdot \rev_1(\dist) \geq \frac{1-\quant}{1-\quant^*}\cdot \rev_1(\dist)
    \end{align*}
    where $\quant^* \triangleq \argmax_{\quant'\in[0, 1]}\revcurve(\quant')$ is the monopoly quantile, and we use the fact that $\revcurve(1) \geq 0$ and $\revcurve(\quant^*) = \rev_1(\dist)$ \citep{mye-81}. After simplification, we obtain
    \begin{align*}
        \quant \geq 1 - \quant
    \end{align*}
    which implies $\quant \geq 1/2$. This finishes the proof of the lemma statement as desired.
\end{proof}

\begin{proof}[Proof of \Cref{pr:4brev_geq_srev}]
Let $\dist = \Pi_{j\in[m]}\dist_j$ be the regular value distributions. Construct auxiliary distribution $\dist\primed = \Pi_{j\in[m]}\dist_j\primed$, where each marginal distribution $\dist_j\primed$ is the two-point distribution constructed in \Cref{lem:brev approx worst dist} using distribution $\dist_j$. By construction, for each item $j\in[m]$, distribution $\dist_j\primed$ is first-order stochastically dominated by distribution $\dist_j$. Consequently, the total value of the grand bundle drawn from distribution $\dist\primed$ is also first-order stochastically dominated by the total value of the grand bundle drawn from distribution $\dist$. Therefore, the revenue monotonicity of pricing the grand bundle implies that
\begin{align*}
    \brev_1(\dist) \geq \brev_1(\dist\primed)
\end{align*}
and thus it suffices to show 
\begin{align*}
    4\cdot \brev_1(\dist\primed) \geq \srev_1(\dist)
\end{align*}
To see this, consider selling the grand bundle at the price equal to $\srev_1(\dist)/2$ to a single bidder with value distribution $\dist\primed$. Its revenue can be expressed as
\begin{align*}
    &
    \frac{1}{2}\cdot \srev_1(\dist) \cdot \prob[\val\sim\dist\primed]{\sum\nolimits_{j\in[m]}\val_j \geq \frac{1}{2}\cdot \srev_1(\dist)}
    \\
    ={} &
    \frac{1}{2}\cdot \srev_1(\dist) \cdot \prob[\val\sim\dist\primed]{\sum\nolimits_{j\in[m]}\val_j \geq \sum\nolimits_{j\in[m]}\frac{1}{2}\cdot \rev_1(\dist_j)}
    \\
    ={} &
    \frac{1}{2}\cdot \srev_1(\dist) \cdot \frac{1}{2} = \frac{1}{4} \cdot \srev_1(\dist)
\end{align*}
where the first equality holds due to the definitions of $\srev_1(\dist)$ and $\rev_1(\dist_j)$, the second equality holds since the purchase probability is equal to $1/2$ (by construction, each distribution $\dist_j\primed$ is a two-point distribution with equal probability at $\rev(\dist_j)$ and 0). Putting the two pieces together, we obtain $4\cdot \brev_1(\dist) \geq 4\cdot \brev_1(\dist\primed) \geq \srev_1(\dist)$ as desired.
\end{proof}

\begin{proof}[Proof of \Cref{cor:18brev_geq_rev}]
Invoking \Cref{thm:six_approx,pr:4brev_geq_srev}, we obtain 
\begin{align*}
    \rev_1 \leq 4\cdot\srev_1 + 2\cdot\brev_1 \leq 4\cdot 4 \cdot \brev_1 + 2\cdot \brev_1 = 18\cdot \brev_1
\end{align*}
which completes the proof of the theorem statement as desired. 
\end{proof}

\xhdr{Improved approximations for MHR items.} When all items are MHR, we obtain the following improved approximation guarantees for the optimal pricing of the grand bundle and the optimal pricing to selling items separately. Their proofs are mainly based on the previous results in \citet{BMP-63,DRY-15}.

\begin{proposition}[Approximation of {\brev} against {\wel} for MHR items and a single bidder]
\label{prop:brev approx wel mhr}
    For $n = 1$ bidder and $m \geq 1$ asymmetric MHR items, the optimal expected revenue from selling all items as a grand bundle is an $e$-approximation to the optimal welfare, i.e., $e\cdot \brev_1 \geq \wel_1$. 
\end{proposition}
\begin{proof}
    By Theorem 3.2 of \citet{BMP-63}, the summation of values drawn from independent but possibly asymmetric MHR distributions is MHR. By Lemma~3.10 of \citet{DRY-15}, the Bayesian optimal revenue of selling a single item to a single bidder with a MHR distribution is an $e$-approximation to the optimal welfare. Combining the two pieces finishes the proof of the proposition as desired.
\end{proof}

\begin{proposition}[Approximation of {\srev} against {\wel} for MHR items and a single bidder]
\label{prop:srev approx wel mhr}
    For $n = 1$ bidder and $m \geq 1$ asymmetric MHR items, the optimal expected revenue from selling all items separately is an $e$-approximation to the optimal welfare, i.e., $e\cdot \srev_1 \geq \wel_1$. 
\end{proposition}
\begin{proof}
    By Lemma~3.10 of \citet{DRY-15}, the Bayesian optimal revenue of selling a single item to a single bidder with a MHR distribution is an $e$-approximation to the optimal welfare. Note that both the welfare benchmark {\wel} and the optimal revenue from selling items separately {\srev} are additive across items. Combining the two pieces finishes the proof of the proposition as desired.
\end{proof}

\subsection{Constant Approximation by the Bundle-based Second-price Auction}  
\label{subsec:bspa approx}
In this section, we study the bundle-based second-price auction ({\bspa}) for two bidders. We begin by focusing on the setting with regular items. Our main result is as follows.
\begin{theorem}[Approximation of {$\bspa_2$} against {$\rev_1$} for regular items]
\label{prop:bvcg constant approx}
    For $m \geq 1$ asymmetric regular items, the expected revenue from the bundle-based second-price auction for two bidders is a 48-approximation to the Bayesian optimal revenue, i.e., $48\cdot \bspa_2 \geq \rev_1$.
\end{theorem}

We first sketch the proof of \Cref{prop:bvcg constant approx}. 
Similar to the approach in \Cref{subsec:brev approx}, our analysis builds on another revenue approximation guarantee established by \citet{BILW-14} (\Cref{thm:core_tail}), which upper bounds the Bayesian optimal revenue by the combination of the optimal revenue from selling items separately ({\srev}) and the revenue from a ``core'' component ({\core}), as defined in \Cref{thm:core_tail}. We show that when the value distributions are regular, {\bspa} is an 8-approximation to {\srev} (\Cref{pr:8bspa_geq_srev}), and a 6-approximation to {\core} when {\core} dominates {\srev}. \Cref{prop:bvcg constant approx} then follows immediately from \Cref{thm:core_tail,pr:8bspa_geq_srev}.

\begin{lemma}[\citealp{BILW-14}] \label{thm:core_tail}
For $n = 1$ bidder and $m \geq 1$ asymmetric items with value distribution $\dist$, define core $\core_1(\dist)$ as
\begin{align*}
    \core_1(\dist)\triangleq \expect[\val\sim\dist]{\sum\nolimits_{j\in[m]}\val_j\given \val_j\leq \srev_1(\dist),\forall j}
\end{align*}
The following property is satisfied
\begin{itemize}
    \item If $\core_1(\dist) \leq 4\cdot \srev_1(\dist)$, then
    \begin{align*}
        \rev_1(\dist) \leq 6\cdot \srev_1(\dist)
    \end{align*}
    \item Otherwise ($\core_1(\dist) > 4\cdot \srev_1(\dist)$), then
    \begin{align*}
        \prob[\val\sim\dist]{\sum\nolimits_{j\in[m]}\val_j \geq \frac{2}{5}\cdot \core_1(\dist)} \geq \frac{47}{72}
    \end{align*}
    and 
    \begin{align*}
        \rev_1(\dist) \leq 2 \cdot \srev_1(\dist) + \core_1(\dist)
    \end{align*}
\end{itemize}
\end{lemma}

\begin{restatable}[Approximation of {$\bspa_2$} against {$\srev_1$} for regular items]{lemma}{lembspaapxgeqsrev}
\label{pr:8bspa_geq_srev}
For $m \geq 1$ asymmetric regular items, the expected revenue from the bundle-based second-price auction for two bidders is an 8-approximation to the optimal expected revenue from selling all items separately for a single bidder, i.e., $8\cdot \bspa_2 \geq \srev_1$.
\end{restatable}
The proof of \Cref{pr:8bspa_geq_srev} closely follows that of \Cref{pr:4brev_geq_srev}, which is included in \Cref{apx:8bspa_geq_srev} for completeness.

\begin{proof}[Proof of \Cref{prop:bvcg constant approx}]
    Consider core $\core_1$ defined in \Cref{thm:core_tail}, and its two cases.

    If core $\core_1 \leq 4\cdot \srev_1$, \Cref{thm:core_tail} ensures that
    \begin{align}
        \rev_1\leq 6\cdot \srev_1 \leq 6\cdot 8\bspa_2 = 48\cdot \bspa_2
    \end{align}
    where the second inequality holds due to \Cref{pr:8bspa_geq_srev}.

    If core $\core_1 > 4\cdot \srev_1$, \Cref{thm:core_tail} ensures that
    \begin{align*}
         \prob[\val\sim\dist]{\sum\nolimits_{j\in[m]}\val_j \geq \frac{2}{5}\cdot \core_1} \geq \frac{47}{72}
    \end{align*}
    and thus the expected revenue of $\bspa_2$ can be lower bounded by the expected payment when both bidders' total values are at least $2\cdot \core_1/5$, i.e.,
    \begin{align*}
        \bspa_2 \geq \frac{2}{5}\cdot \core \cdot 
        \left(\prob[\val\sim\dist]{\sum\nolimits_{j\in[m]}\val_j \geq \frac{2}{5}\cdot \core_1}\right)^2 = \frac{2209}{12960}\cdot \core_1
    \end{align*}
    In addition, \Cref{thm:core_tail} upper bounds the Bayesian optimal revenue as
    \begin{align*}
        \rev_1 \leq 2\cdot \srev_1 + \core_1 
        \leq \left(2\cdot 8 + \frac{12960}{2209}\right)\cdot\bspa_2 \leq 22\cdot \bspa_2
    \end{align*}
    which finishes the proof of \Cref{prop:bvcg constant approx} as desired. 
\end{proof}

\xhdr{Improved approximations for MHR items.} When all items are MHR, we obtain the following improved approximation guarantees for the bundle-based second-price auction and the VCG auction. Their proofs closely follow those of \Cref{prop:brev approx wel mhr,prop:srev approx wel mhr}, which are included in \Cref{apx:propbspawelapprox,apx:propvcgwelapprox} for completeness.

\begin{restatable}[Approximation of $\bspa_2$ against $\wel_1$ for MHR items]{proposition}{propbspawelapprox}
\label{prop:bspa approx wel mhr}
    For $m \geq 1$ asymmetric MHR items, the expected revenue from the bundle-based second-price auction for two bidders is an $e$-approximation to the optimal welfare, i.e., $e\cdot \bspa_2 \geq \wel_1$.
\end{restatable}

\begin{restatable}[Approximation of $\vcg_2$ against $\wel_1$ for MHR items]{proposition}{propvcgwelapprox}
\label{prop:vcg approx wel mhr}
    For $m \geq 1$ asymmetric MHR items, the expected revenue from the VCG auction for two bidders is an $e$-approximation to the optimal welfare, i.e., $e\cdot \vcg_2 \geq \wel_1$.
\end{restatable}

\section{Conclusion and Future Directions}
\label{sec:conclusion}

In this work, we study the \emph{competition complexity} in multi-item auction settings. For the Vickrey--Clarke--Groves ({\vcg}) auction, we identify its \emph{tight} competition complexity, denoted by $\constna$, when item valuations follow $\alpha$-strongly regular distributions, including the MHR case (i.e., $\alpha = 1$). In contrast to prior work that assumes only regularity and yields bounds depending on both the number of items $m$ and bidders $n$, our bound $\constna$ depends solely on $n$ and is independent of $m$.
We then turn to the bundle-based second-price auction ({\bspa}) in the single-bidder setting. For MHR distributions, we show that {\bspa} has a competition complexity of 3. For regular distributions, we derive competition complexity bounds in cases with two or three items. In addition, we establish constant-factor approximation guarantees for both {\bspa} and the revenue from selling items as a grand bundle ({\brev}), under regular and MHR assumptions.
Many of our techniques are general and may be of independent interest.

\xhdr{Future directions.} There are many interesting directions for future research. 

First, it would be valuable to extend the competition complexity results of {\bspa} from the case of two or three items to a general number $ m $ of items. Note that most parts of our analysis framework (\Cref{prop:cdwapprox,coro:cdwapprox,prop:bspa_min_max}) hold for arbitrary $ m $. Therefore, one promising approach to establishing a competition complexity of $ m $ is to generalize \Cref{prop:cc_bspa_m_3}. Going further, an important open question is whether {\bspa} admits a competition complexity with sub-logarithmic or even constant dependence on $m$ in the single-bidder setting.

Second, from a technical perspective, a key challenge in analyzing {\bspa} is the limited understanding of the distribution of sums of values drawn from regular distributions. By comparison, \citet{BMP-63} showed that sums of values from MHR distributions remain MHR. It would be valuable to prove or disprove an analogous result for regular distributions, even in the i.i.d.\ setting.

Third, as shown in \Cref{bspa_multiple_n}, for two or more bidders, the competition complexity of {\bspa} suffers from exponential dependence on $m$, even under uniform valuations. Hence, it would be valuable to design prior-independent mechanisms with improved competition complexity guarantees in the multi-bidder setting. One promising direction could be a hybrid mechanism combining {\vcg} and {\bspa}, for example, by first running a competition for entry and then selling items separately.

Finally, it would be interesting to improve the constant-factor approximation or identify tight bounds for {\bspa} and {\brev} in the single-bidder, regular-items setting.

\subsection*{Acknowledgments} 
Yingkai Li thanks the NUS Start-up Grant for financial support. S.\ Matthew Weinberg's research is supported by NSF CAREER CCF-1942497. 
We are grateful to Jason Hartline for invaluable discussions in the early stages of this project and to the anonymous reviewers of EC 2025 for their time and effort in carefully reading this paper.

\bibliography{auctions}

\appendix

\section{Motivation for the Bundle-based Second-price Auction}
\label{apx:bspa motivation}
In this section, we provide some motivation why the bundle-based second-price auction ({\bspa}) could be an interesting mechanism for studying competition complexity in the multi-item setting.

First, selling items as a grand bundle is a simple and practical approach in the multi-item setting. \citet{BILW-14} has identified selling items separately ({\srev}) and selling items as a grand bundle ({\brev}) as two simple mechanisms that complement each other in the few-bidders setting that we study. The Bulow-Klemperer-style questions typically emphasize prior-independent mechanisms, and ({\bspa}) is a natural prior-independent implementation of selling the grand bundle (analogous to {\vcg} vs {\srev}).

As illustrated numerically in \Cref{example:intro:eqr} and discussed in \citet{DRWX-24}, under the equal-revenue distribution (an ``extreme case'' within the class of regular distributions), {\bspa} outperforms {\vcg} and closely matches the Bayesian optimal revenue. Specifically, \citet{DRWX-24} shows that when selling $m$ items with equal-revenue distributions to $n$ bidders, the expected revenue of {\bspa} is $n m + \Omega(m\ln(m n))$, while the Bayesian optimal revenue is $n m + O(m^2 \ln n)$. As a comparison, the expected revenue of {\vcg} for the equal-revenue distribution is $nm$.
This further motivates {\bspa} as a mechanism of interest in the Bulow-Klemperer setting, which typically focuses on {\vcg} in this literature.

\section{Omitted Proof of Lemma~\ref{prop:cc_bspa_m_3} with Three Bidders} 
\label{app:sec:cc_bundle}

\propBspaThree*

\begin{proof}[Proof of \Cref{prop:cc_bspa_m_3} with $m = 3$ bidders.]
    Here we consider the scenario where $m = 3$ and $\quantMatrix$ is a $3 \times 4$ matrix. Again let $\quantMatrix^*$ be an associated matrix where $\quantMatrix^*[i, j]$ is the $j$-th highest order statistic among $\quantMatrix[i, 1], \cdots, \quantMatrix[i, m+1]$. We will propose a specific strategy for the optimizing player in the game $G_{\dist, \quantMatrix}$, and show that for all adversary strategy the optimizer's expected payoff is at least $\cdw_1(\quantMatrix)$. Without loss of generality, we assume the adversary strategy $A$ is always an extreme point of $\Delta([m+1])$, since there always exists an optimal adversary strategy of this form (as witnessed in the proof of \Cref{prop:bspa_min_max}).
    
    We consider the following mutually exclusive events separately: 
    \begin{enumerate}
        \item 
        After $\quantMatrix$ is permuted by $\sigma$, there is a column $j^*$ that contains only the first or second highest order statistics in all $3$ rows. Formally, there exists a $j^* \in [4]$ such that for all rows $i \in [3]$, $\quantMatrix[i, \sigma_{i}(j^*)] = \quantMatrix^*[i, 1]$ or $\quantMatrix^*[i, 2]$. 
        \item 
        When the first condition is not true but the first and second row order statistics in all rows occupy only three columns. Formally, when there exists a unique column $j^*$ such that for all $i \in [3]$, $\quantMatrix[i, \sigma_{i}(j^*)] = \quantMatrix^*[i, 3]$ or $\quantMatrix^*[i, 4]$.
        \item 
        When the first and second conditions are not true, but when $2$ columns contain two first and second row order statistics. Formally, when there exists two columns $j^*_1$ and $j^*_2$ such that for each $k \in \{1, 2\}$, exists $i_1$ and $i_2$ such that $\quantMatrix[i_1, \sigma_{i_1}(j^*_k)] = \quantMatrix^*[i_1, 1]$ or  $\quantMatrix^*[i_1, 2]$ and $\quantMatrix[i_2, \sigma_{i_2}(j^*_k)] = \quantMatrix^*[i_2, 1]$ or  $\quantMatrix^*[i_2, 2]$. 
        \item 
        When none of the above three conditions hold. 
    \end{enumerate}
    An illustration of all four cases above can be found in \Cref{fig:bspa_m_3_cases}. We first calculate the probability of the four cases. Let $j_1$ and $j_2$ be the two columns such that $\quantMatrix[1, \sigma_{1}(j_1)] = \quantMatrix^*[1, 1]$ and $\quantMatrix[1, \sigma_{1}(j_2)] = \quantMatrix^*[1, 2]$. 

    \begin{itemize}
        \item 
        For case $1$,  the probability of $j_1$ satisfying the condition (contain only first/second row order statistic) is equal to $2/4 \cdot 2/4 = 1/4$. Similar the probability of $j_2$ satisfying the condition is $1/4$. Any other column cannot satisfy the condition since they don't contain the first/second row order statistics for row $1$. The probability that both $j_1$ and $j_2$ satisfy the condition is $1/36$. Hence by inclusion exclusion principle, the probability of case $1$ is $1/4 + 1/4 - 1/36 = \frac{17}{36}$.
        \item 
        For case $2$, since the first condition (there exists a column that contain only first/second row order statistics) does not hold, and all the first/second row order statistics belong to 3 columns. This means that each column has two first/second row order statistics. We can calculate that this happens with probability $1/9$. 
        \item 
        For case $3$, it must be the case that two rows share the columns that contain their first and second row order statistics, while the remaining row has a disjoint set of columns that contain their first and second row order statistics. This occurs with probability $3 \cdot 1/6 \cdot 1/6 = 1/12$. 
        \item 
        For case $4$, this happens with remaining probability of $1 - (1/2 - 1/36) - 1/9 - 1/12 = 1/3$.
    \end{itemize}

    We now specify a fixed strategy for the optimizer in the two-player zero-sum game defined in \Cref{def:two player zero sum game}. For each of the four cases, we define the optimizer's choice of column $j_p$ as a function of $\sigma$, $\tau$, and the adversary's strategy. This allows us to derive a lower bound on $\bspa_{4}(\quantMatrix)$.  

    \begin{itemize}
        \item 
        In case $1$, if the adversary does not pick the column $j^*$ that contains only first and second row order statistics, then the optimizer can pick column $j^*$ where $\quantMatrix[\tau(i), \sigma_{\tau(i)}(j^*)] \geq \quantMatrix^*[\tau, 2]$ with probability $1$ for each item $i$. Otherwise, the optimizer could pick $j_p$ with the largest number of first and second row order statistics (break tie towards the column that does not contain $4^{th}$ order statistic $\quantMatrix^*[\cdot, 4]$ for any row, otherwise break tie evenly). The probability that $j_p$ contain $1$, $2$ and $3$ first/second row order statistic is $2/9$, $2/3$ and $1/9$ respectively, and these probabilities are symmetry regarding the row index (for any item $i$, probability of selecting $\geq \quantMatrix^*[\tau(i), 2]$ is the same). Hence over all permutations $\sigma$ that satisfy the first condition, on average for each item $i$, the optimizer selects a quantile $\geq \quantMatrix^*[\tau(i), 2]$ with probability at least $2/9 \cdot 1/3 + 2/3 \cdot 2/3 + 1/9 \cdot 1= \frac{17}{27}$. 
        
        For tie breaking, notice that there are in total $6$ first/second row order statistics, and column $j^*$ contains $3$ of them, so the remaining columns contain $3$ of first/second row order statistics in total. This means that tie breaking is only necessary when each column other than $j^*$ contains exactly one first/second row order statistic (let this be Event $\textsf{ONE}$). When tie breaking is not necessary, the column optimizer select is entirely dependent on the positions of first/second row order statistics. Hence for any item $i$, the selected column contain $\quantMatrix^*[\tau(i), 3]$ and $\quantMatrix^*[\tau(i), 4]$ with equal probability due to $\sigma$ being uniformly random. Conditioned on Event $\textsf{ONE}$, the probability that all columns contain $4^{th}$ order statistics is $\frac{1}{4}$, and the probability otherwise is $\frac{3}{4}$. In the first case, for all item $i$, probability of picking a column that contain $\quantMatrix^*[\tau(i), 3]$  and $\quantMatrix^*[\tau(i), 4]$ are both equal to $1/3$. In the second case, for all item $i$, probability of picking a column that contain $\quantMatrix^*[\tau(i), 3]$ is $2/3$, and probability of picking a column that contain $\quantMatrix^*[\tau(i), 4]$ is $0$. Thus overall, the optimizer selects a column that contains quantile $\quantMatrix^*[\tau(i), 3]$ with probability $2/9 \cdot (3/4 \cdot 2/3 +1/4 \cdot 1/3) + 2/3 \cdot 1/6 = 13/54$, and the optimizer selects a column that contains quantile $\quantMatrix^*[\tau(i), 4]$ with probability $2/9 \cdot (3/4 \cdot 0 +1/4 \cdot 1/3) + 2/3 \cdot 1/6 = 7/54$. 
        % The optimizer selects the quantile $\quantMatrix^*[i, 3]$ and $\quantMatrix^*[i, 4]$ with equal remaining probability: $\frac{5}{27}$. 
        \item 
        In case $2$, let $j^*$ be the column that does not contain any first/second order statistics. Let the optimizer's strategy to be: no matter what the adversary does, the optimizer selects a remaining column that contains first/second order statistics with equal probability. For any item $i$, if $\quantMatrix[\tau(i), j^*] = \quantMatrix[\tau(i), 3]$, then the distribution of $\quantMatrix[\tau(i), j_p]$ stochastically dominates $\frac{1}{2} \cdot \quantMatrix[\tau(i), 2] + \frac{1}{2} \cdot \quantMatrix[\tau(i), 4]$. Otherwise, if $\quantMatrix[\tau(i), j^*] = \quantMatrix[\tau(i), 4]$, then the distribution of $\quantMatrix[\tau(i), j_p]$ stochastically dominates $\frac{1}{2} \cdot \quantMatrix[\tau(i), 2] + \frac{1}{2} \cdot \quantMatrix[\tau(i), 3]$. Since across all permutations $\sigma$ and $\tau$, in the second case, $\quantMatrix[\tau(i), j^*] = \quantMatrix[i\tau(i), 3]$ and $\quantMatrix[\tau(i), j^*] = \quantMatrix[\tau(i), 4]$ occur with equal probability, the distribution of $\quantMatrix[\tau(i), j_p]$ conditioned on the permutation $\sigma$ is in the second case stochastically dominates $1/2 \cdot \quantMatrix[\tau(i), 2] + 1/4 \cdot \quantMatrix[\tau(i), 3] + 1/4 \cdot \quantMatrix[\tau(i),4]$.  

        \item 
        In case $3$, there are two columns $j^*_1$ and $j^*_2$ such that for two rows $i_1$ and $i_2$, $\quantMatrix[i_1, \sigma_{i_1}(j^*_k)] \geq \quantMatrix^*[i_1, 2]$ and $\quantMatrix[i_2, \sigma_{i_2}(j^*_k)] \geq \quantMatrix^*[i_2, 2]$. No matter which column the adversary picks, the optimizer can pick the remain column between $j^*_1$ and $j^*_2$. Due to symmetry across all permutations $\sigma$ and $\tau$, for any item $i$, the optimizer selects a quantile $\quantMatrix[\tau(i), j_p]$ that stochastically dominates $2/3 \cdot \quantMatrix^*[\tau(i), 2] + 1/3 \cdot \quantMatrix^*[\tau(i), 4]$. 
        \item 
        In case $4$, the optimizer will just select a uniformly random remaining column that is not selected by the adversary. Hence for all item $i$, the the distribution of $\quantMatrix[\tau(i), j_p]$ stochastically dominates $1/3 \cdot \quantMatrix[i, 2] + 1/3 \cdot \quantMatrix[i, 3] + 1/3 \cdot \quantMatrix[i,4]$.
    \end{itemize} 
    % the probability of this happening is equal to the probability that 1) exactly one $j^* \in \{\}$ satisfy $\quantMatrix[\tau(2), \sigma_{\tau(2)}(j)]$

    We conclude that in general, the distribution of $\quantMatrix[\tau(i), j_p]$ stochastically dominates 
    \begin{align*}
        &\left(\frac{17}{36} \cdot \frac{17}{27} + \frac{1}{9} \cdot \frac{1}{2} + \frac{1}{12} \cdot \frac{2}{3} + \frac{1}{3} \cdot \frac{1}{3}\right) \cdot \quantMatrix^*[\tau(i), 2] 
        + \left(\frac{17}{36} \cdot \frac{13}{54} + \frac{1}{9} \cdot \frac{1}{4} + \frac{1}{3} \cdot \frac{1}{3}\right) \cdot \quantMatrix^*[\tau(i), 3] \\
        &\hfill + \left(\frac{17}{36} \cdot \frac{7}{54} + \frac{1}{9} \cdot \frac{1}{4} + \frac{1}{12} \cdot \frac{1}{3} + \frac{1}{3} \cdot \frac{1}{3}\right) \cdot \quantMatrix^*[\tau(i), 4]\\
        &= 505/972 \cdot \quantMatrix^*[\tau(i), 2] + 491/1944 \cdot \quantMatrix^*[\tau(i), 3] + 443/1944 \cdot \quantMatrix^*[\tau(i), 4], 
    \end{align*}
    which stochastically dominates $1/2 \cdot \quantMatrix[\tau(i), 2] + 1/4 \cdot \quantMatrix[\tau(i), 3] + 1/4 \cdot \quantMatrix[\tau(i), 4]$, i.e., the $i$-th component of $\cdw_1(\quantMatrix)$. This completes the proof of the lemma statement with $m = 3$ bidders as desired.
\end{proof}

\begin{figure} 
\begin{center}
\begin{tikzpicture}

% Case 1
\draw (2,-1) node[anchor=south] {Case 1};
\foreach \x in {0,1,2,3} \foreach \y in {0,1,2} \draw (\x,\y) rectangle (\x+1,\y+1);
\node at (0.5,2.5) {x};
\node at (0.5,1.5) {x};
\node at (0.5,0.5) {x};

% Case 2
\begin{scope}[xshift=5cm]
\draw (2,-1)  node[anchor=south] {Case 2};
\foreach \x in {0,1,2,3} \foreach \y in {0,1,2} \draw (\x,\y) rectangle (\x+1,\y+1);
\node at (1.5,2.5) {x};
\node at (2.5,2.5) {x};
\node at (0.5,1.5) {x};
\node at (1.5,1.5) {x};
\node at (0.5,0.5) {x};
\node at (2.5,0.5) {x};
\end{scope}

% Case 3
\begin{scope}[yshift=-5cm]
\draw (2,-1)  node[anchor=south] {Case 3};
\foreach \x in {0,1,2,3} \foreach \y in {0,1,2} \draw (\x,\y) rectangle (\x+1,\y+1);
\node at (0.5,2.5) {x};
\node at (1.5,2.5) {x};
\node at (0.5,1.5) {x};
\node at (1.5,1.5) {x};
\node at (2.5,0.5) {x};
\node at (3.5,0.5) {x};
\end{scope}

% Case 4
\begin{scope}[xshift=5cm,yshift=-5cm]
\draw (2,-1)  node[anchor=south] {Case 4};
\foreach \x in {0,1,2,3} \foreach \y in {0,1,2} \draw (\x,\y) rectangle (\x+1,\y+1);
\node at (0.5,2.5) {x};
\node at (1.5,2.5) {x};
\node at (0.5,1.5) {x};
\node at (2.5,1.5) {x};
\node at (2.5,0.5) {x};
\node at (3.5,0.5) {x};

%\foreach \x in {0,1,2,3} \foreach \y in {0,1,2} \draw (\x,\y) rectangle (\x+1,\y+1);
\end{scope}

\end{tikzpicture}
\end{center}
\caption{Illustration of the four cases in \Cref{prop:cc_bspa_m_3}, "x" represent the position of first/second row order statistics. Case 1: there are three "x" in a column; Case 2: there exists a unique column without "x"; Case 3: there is a $4 \times 4$ sub-grid with "x"; Case 4: remaining configurations. } \label{fig:bspa_m_3_cases}
\end{figure}

\section{Omitted Proofs in Section~\ref{sec:approx_bundle}} 
\label{apx:brev-approx}
In this section, we include all the omitted proofs in \Cref{sec:approx_bundle}.

\subsection{Omitted Proof for Lemma~\ref{pr:8bspa_geq_srev}}
\label{apx:8bspa_geq_srev}

\lembspaapxgeqsrev*

\begin{proof}[Proof of \Cref{pr:8bspa_geq_srev}]
    Let $\dist = \Pi_{j\in[m]}\dist_j$ be the regular value distributions. Construct auxiliary distribution $\dist\primed = \Pi_{j\in[m]}\dist_j\primed$, where each marginal distribution $\dist_j\primed$ is the two-point distribution constructed in \Cref{lem:brev approx worst dist} using distribution $\dist_j$. By construction, for each item $j\in[m]$, distribution $\dist_j\primed$ is first-order stochastically dominated by distribution $\dist_j$. Consequently, the total value of the grand bundle drawn from distribution $\dist\primed$ is also first-order stochastically dominated by the total value of the grand bundle drawn from distribution $\dist$. Therefore, the revenue monotonicity of the second-price auction implies that
\begin{align*}
    \bspa_2(\dist) \geq \bspa_2(\dist\primed)
\end{align*}
and thus it suffices to show 
\begin{align*}
    8\cdot \bspa_2(\dist\primed) \geq \srev_1(\dist)
\end{align*}
To see this, consider the total value threshold of $\srev_1(\dist)/2$, by construction, the probability that one bidder's total value exceed this threshold is 
\begin{align*}
    &
    \prob[\val\sim\dist\primed]{\sum\nolimits_{j\in[m]}\val_j \geq \frac{1}{2}\cdot \srev_1(\dist)}
    \\
    ={} &
    \prob[\val\sim\dist\primed]{\sum\nolimits_{j\in[m]}\val_j \geq \sum\nolimits_{j\in[m]}\frac{1}{2}\cdot \rev_1(\dist_j)}
    \\
    ={} &
    \frac{1}{2}
\end{align*}
where the first equality holds due to the definitions of $\srev_1(\dist)$ and $\rev_1(\dist_j)$, the second equality holds since the purchase probability is equal to $1/2$ (by construction, each distribution $\dist_j\primed$ is a two-point distribution with equal probability at $\rev(\dist_j)$ and 0). 

Therefore, we can lower bound the expected revenue of $\bspa_2(\dist\primed)$ by the expected payment when both bidders' total values are at least $\srev_1(\dist)/2$, i.e.,
\begin{align*}
    \bspa_2(\dist\primed) \geq \frac{1}{2}\cdot \srev_1(\dist) \cdot 
    \left(\prob[\val\sim\dist\primed]{\sum\nolimits_{j\in[m]}\val_j \geq \frac{1}{2}\cdot \srev_1(\dist)}\right)^2 = \frac{1}{8}\cdot \srev_1(\dist)
\end{align*}
Putting the two pieces together, we obtain $8\cdot \bspa_2(\dist) \geq 8\cdot \bspa_2(\dist\primed) \geq \srev_1(\dist)$ as desired.
\end{proof}

\subsection{Omitted Proof for Proposition~\ref{prop:bspa approx wel mhr}}
\label{apx:propbspawelapprox}
\propbspawelapprox*
\begin{proof}
    By Theorem 3.2 of \citet{BMP-63}, the summation of values drawn from independent but possibly asymmetric MHR distributions is MHR. Hence, invoking \Cref{thm:BK}, the expected revenue from the bundle-based second-price auction with two bidders is at least the optimal revenue of selling all items as a grand bundle to a single bidder, i.e., $\bspa_2 \geq \brev_1$. Invoking \Cref{prop:brev approx wel mhr} finishes the proof as desired.
\end{proof}

\subsection{Omitted Proof for Proposition~\ref{prop:vcg approx wel mhr}}
\label{apx:propvcgwelapprox}
\propvcgwelapprox*
\begin{proof}
    Invoking \Cref{thm:BK}, the expected revenue from the VCG auction with two bidders is at least the optimal revenue of selling items separately to a single bidder, i.e., $\vcg_2 \geq \srev_1$. Invoking \Cref{prop:srev approx wel mhr} finishes the proof as desired.
\end{proof}

\end{document}